\documentclass[journal]{IEEEtran}
\usepackage{cite}

  \usepackage[pdftex]{graphicx}
  \graphicspath{{./figures/}{.png}}
\usepackage{amsmath, amssymb, amsthm, amsbsy, mathtools, mathptmx}
\usepackage[mathscr]{euscript}

\newtheorem{definition}{Definition}
\newtheorem{proposition}{Proposition}
\newtheorem{corollary}{Corollary}
\newtheorem{lemma}{Lemma}
\newtheorem{theorem}{Theorem}
\DeclareMathOperator*{\argmax}{arg\,max}
\DeclareMathOperator*{\argmin}{arg\,min}
\newcommand\eqnumber{\addtocounter{equation}{1}\tag{\theequation}}

\newcommand{\de}{\delta}
\newcommand{\hV}{\mathscr{V}}
\newcommand{\hA}{\mathscr{A}}
\newcommand{\hT}{\mathscr{T}}
\newcommand{\hF}{\mathscr{F}}
\newcommand{\hS}{\mathscr{S}}
\newcommand{\hI}{\mathscr{I}}
\newcommand{\ie}{\textit{i.e.,}}
\newcommand{\eg}{\textit{e.g.,}}
\newcommand{\rr}{r}
\newcommand{\br}{\bar r}

\usepackage{algorithm, algorithmic}

\usepackage{multirow}
\usepackage[table]{xcolor}
\usepackage{tabularx}

\usepackage{caption, subcaption}
\usepackage{url, booktabs, hyperref}
\usepackage{hyperref}
\hypersetup{
    colorlinks=true,
    allcolors=blue}
\usepackage[nameinlink,capitalise]{cleveref}

\Crefname{appsec}{Appendix}{Appendices}

\begin{document}
%
\title{Smoothed Least-Laxity-First Algorithm for EV Charging}
%
%

\author{Niangjun~Chen,~\IEEEmembership{Member,~IEEE,}
Christian~Kurniawan,
Yorie~Nakahira,~\IEEEmembership{Member,~IEEE,}
    Lijun~Chen,~\IEEEmembership{Member,~IEEE,}
        Steven~H.~Low,~\IEEEmembership{Fellow,~IEEE}
\thanks{N. Chen is with the Department of Information Systems Technology and Design, Singapore University of Technology and Design, Singapore 487372. N. Chen has join appointment with the Institute for High Performance Computing, Agency for Science, Technology, and Research, Singapore 138632. (e-mail: niangjun\_chen@sutd.edu.sg, chennj@ihpc.a-star.edu.sg)}
\thanks{C. Kurniawan and Y. Nakahira are with the Department
of Electrical and Computer Engineering, Carnegie Mellon University, PA 15213, USA. (e-mail: yorie@cmu.edu, christian.paryoto@gmail.com)}
\thanks{L. Chen is with the Department of Computer Science, University of Colorado Boulder, CO, 80309 USA. (e-mail: lijun.chen@colorado.edu)}
\thanks{S. H. Low is with the Department
of Computing and Mathematical Sciences, California Institute of Technology,
CA, 91125 USA. (e-mail: slow@caltech.edu)}}

\maketitle

\begin{abstract}
Adaptive charging can charge electric vehicles (EVs) at scale cost effectively, despite the uncertainty in EV arrivals. We formulate adaptive EV charging as a feasibility problem that meets all EVs' energy demands before their deadlines while satisfying constraints in charging rate and total charging power. We propose an online algorithm, smoothed least-laxity-first (sLLF), that decides the current charging rates without the knowledge of future arrivals and demands. We characterize the performance of the sLLF algorithm analytically and numerically. Numerical experiments with real-world data show that it has a significantly higher rate of feasible EV charging than several other existing EV charging algorithms. Resource augmentation framework is employed to assess the feasibility condition of the algorithm. The assessment shows that the sLLF algorithm achieves perfect feasibility with only a 0.07 increase in resources.
\end{abstract}


%
\IEEEpeerreviewmaketitle

\section{Introduction}
\label{S:Introduction}
%
%
%
%

\IEEEPARstart{T}{he} electrification of transportation provides an important opportunity for energy efficiency and sustainability. There were over seven million of pure and hybrid electric vehicles (EVs) worldwide as of 2019~\cite{EVoutlook}, and EV proliferation is expected to accelerate for many years to come. EV charging at scale, however, presents a tremendous challenge as uncontrolled EV charging may strain the power grid and cause voltage instability. One way to mitigate the impact and stabilize the power grid as well as to manage uncertainty in the energy supply from renewable energy resources such as wind power and solar energy is by exploiting the flexibility in charging time and rate. To exploit and optimize this flexibility, many EV charging algorithms have been proposed.

There is a very large literature on EV charging algorithms and they can be categorized as either offline or online. The offline algorithms require complete information on all EVs to decide the charging rates, \eg~\cite{sundstrom2010planning,6670091,richardson2012local,gan2013optimal, ma2010decentralized, chen2014electric}. Yet, information on future EV arrivals may not be available or expensive to obtain, which motivates the development of online algorithms, \eg~\cite{gan2013real, chen2012large, yu2016deadline,Stein:2012:MOM:2343776.2343792,sha2004real,6670091,carvalho2015congestion,7809844,Guo}. However, an online algorithm, which uses only information from EVs present at the charging station to decide their charging rates, may not produce a solution that satisfies all the constraints even when all EVs' demands can be satisfied. Thus, the efficacy of these online algorithms still depends on the accurate prediction of EV arrivals and energy demands that is difficult to obtain. The optimum charging rate is obtained by solving either a convex optimization (\eg \cite{AlonsoOptimalAlgorithm2014,MaoIntelligentModes2019,ElmehdiGeneticFleet2019}) or a linear programming problem (\eg \cite{Guo}).
To reduce the computational complexity and memory usage, sorting or bisection based methods (\eg earliest-deadline-first, least-laxity-first~\cite{stankovic2012deadline}, and Whittle's index policy~\cite{whittle1988restless,yu2016deadline}) are often employed. Nevertheless, the lack of information on future EV arrivals remains the major challenge for solving the problem. Moreover, these algorithms require temporal coordination across time among a large number of EVs which is hard. 

In view of these limitations, we investigated low-complexity EV charging that does not require prediction of EV arrivals/demands or temporal coordination. We first formulated the charging rate allocation as a feasibility problem \textbf{to satisfy the energy demands of all EVs before their departure under constraints of individual maximum charging rate of every EV and the total available power supply.} We then proposed an online algorithm, the smoothed least-laxity-first (sLLF), based on the classic least-laxity-first (LLF) with an improved success rate in achieving feasibility, that decides on the current charging rates based only on the information up to the current time. Without information on future EV arrivals, the sLLF algorithm makes the best possible decision by maximizing the minimum resulting laxity for the next time among the EVs currently in the system. Here, laxity can be seen as the feasibility margin for EV charging and is defined as the EV's remaining time at the charging station decreased by the time needed to be fully charged at the maximum charging rate. By considering only the EVs up to the current time, an (offline) feasible instance may be (online) infeasible under sLLF. Additionally, unlike the classic LLF algorithm, the sLLF algorithm avoids unnecessary oscillations in the charging rates.

Cost related to the installation, replacement, and development of both the infrastructure of a charging station including power generation and the battery of an EV is also a factor to be considered in a charging algorithm \cite{BrennaElectricEstimation2020}. Generally, the algorithm needs to adhere to the limitation of the resources while still producing a feasible solution \cite{AmjadATechniques2018}. Thus, the feasibility condition of an algorithm can be assessed by characterizing the minimum amount of additional resources (\ie total power supply and charging rates) that will allow the algorithm to produce a feasible solution, assuming all EVs' demands can be satisfied. In this study, the feasibility condition of the sLLF algorithm is analyzed using the resource augmentation framework~\cite{kalyanasundaram2000speed,Phillips2002, im2014competitive, im2014selfishmigrate}. Resource augmentation is a prominent analysis framework for analyzing the performance of online algorithms for multiprocessor scheduling~\cite{liu1973scheduling,dertouzos1989multiprocessor,davis2011survey}. We apply this framework to the EV charging problem that can be viewed as a deadline scheduling problem by considering chargers as processors and EVs with certain energy demand as jobs. Contrary to the traditional application of the framework, in our setting the power limit is time-varying, the maximum rates are heterogeneous, and the power limit may not necessarily be an integer multiplication of the maximum rate. Our work is the first to extend resource augmentation into the cases for heterogeneous processors with a time-varying number.

We further carried out numerical experiments using real-world datasets and showed that sLLF has a significantly higher rate of generating feasible EV charging schedules than several other common EV charging algorithms. This is expected, as the sLLF algorithm tries to leave the largest feasibility margin, so it can best accommodate arbitrary future EV arrivals. The datasets we employed are collected from Google’s facilities in Mountain View (Google\_mtv) and Sunnyvale (Google\_svl) as well as the adaptive charging network (ACN) testbed we deployed at California Institute of Technology (Caltech), called CAGarage. At Caltech ACN, each EV arrives at a charger with energy demand and departure time. The charging facility also has a time-varying total power supply. The ACN performs real-time sensing, communication, and control using the profiles of each EV (including energy demand, departure time, and maximum charging rate) to decide the charging rate of each EV.
See \cite{Lee2016, ZLee2018} for more details on the Caltech infrastructure and \cite{LeeLiLow2019a} on the charging data.

The rest of the paper is organized as follows: \cref{S:Model and Algorithm} introduces the system model and proposes the sLLF algorithm; the performance of the sLLF algorithm is analyzed via the procedure describes in \cref{S:Assessment Procedure}; then the result and discussion of the performance analysis are presented in \cref{S:Performance Evaluation and Comparison}.

\section{Model and Algorithm}
\label{S:Model and Algorithm}

\subsection{System Model}
\label{SS:System Model}

In this study, we consider a system with one charging station that serves a set of EVs, indexed by $i\in\hV=\{1,\ 2,\ 3,\ \cdots \}$. We use a discrete-time model where time is divided into slots of equal sampling intervals, indexed by $t \in \hT=\{ 0,\ 1,\ 2,\ \cdots , T\}$. EV $i$ arrives at the charging station with an energy demand $e_i$ at time $a_i$, and departs from the station at time $d_i$\footnote{Each EV leave at its departure time regardless of its charging conditions. This assumption is applicable for most slow chargers including ACN~\cite{Lee2016}. Under this assumption, we do need to explicitly model the number of stations, as the speed of charging does not affect the availability of chargers for incoming EVs.}. During its stay at the station, the EV is charged at a rate (or power) of $r_i(t) \geq 0,\ a_i\leq t< d_i$. For convenience, since the rate is adjusted at a discrete-time which make the value of the charging rate of vehicle $i$ at a given time is equal to the energy transmitted to the vehicle, we extend the definition of $r_i(t)$ to the entire temporal domain such that the notation $r_i(t)$ can be interpreted as both the charging rate of and the energy transmitted into the vehicle $i$ at time $t$. We also denote the set of the remaining EVs in the charging station at time $t$ as $V_t= \{ i \in A_t : a_i \leq t < d_i \}$ and the remaining energy demand of EV $i$ at time $t$ as $e_i(t)$. The notations are summarized in \cref{table:notations 1}.

Constraints of the system need to be defined to satisfy the charger and power supply limitation along with the vehicle's energy demand. To account for the limitations in the charger or battery of an EV, each EV $i$ can only be charged up to a peak rate, $\br_i$:

\begin{equation}
	\label{eq:rate constraint}
	\begin{aligned}
	& & &
	\begin{cases}
		0 \leq r_i (t)  \leq \br_i , & t \in [a_i  , d_i ), i\in\hV  \\
		r_i (t) = 0, & t \notin [a_i  , d_i ), i\in\hV 
	\end{cases} \\
	& \text{where} & & \\
	& & & \br_{\min} \leq \br_i \leq \br_{\max},\ i\in\hV
	\end{aligned}
\end{equation}

\noindent 
To account for the limitations in the power grid or station, the charging station has a (possibly time-varying) power limit, $P(t)$, such that\footnote{All EVs at the charging station can be simultaneously charged as long as the constraints \cref{eq:rate constraint,eq:power constraint} are satisfied.}

\begin{equation}
	\label{eq:power constraint}
	\begin{aligned}
	& & & \sum_{ i \in \hV} r_i (t)  \leq P(t),\ t\in\hT \\
	& \text{where} & & 0 \leq P_{\min} \leq P(t) \leq P_{\max}
	\end{aligned}
\end{equation}

\noindent
Finally, every EV's energy demands need to be satisfied\footnote{The actual constraint in ACN is $\sum_{t\in\hT} \de r_i (t) = e_i,~i\in \hV $, where $\de$ is the sojourn time in unit of hours of sampling time intervals, $e_i$ has unit of $kWh$, $r_i(t)$ has unit of $kW$~\cite{Lee2016}. Since $r_i(t)$ can always be rescaled according to $\de$, we set $\de = 1$ without loss of generality.}
 
\begin{equation}
	\label{eq:energy demand constraint}
	\sum_{t\in\hT} r_i (t)  = e_i,~i\in \hV
\end{equation}

The charging problem instance is then defined as a quintuple ${\hI}=\{a_i, d_i, e_{i}, \bar{r}_{i}; P(t)\}_{i\in\hV, t\in\hT}$. The primary goal of EV charging is to satisfy every EV's energy demands under the above power supply and peak charging rate constraints (\cref{eq:rate constraint,eq:power constraint,eq:energy demand constraint}). 

\begin{definition}[feasible instance]
	\label{definition:feasible instance}
	An EV charging problem instance ${\hI}$ is \textit{offline feasible} if there exist charging rates $r = \{ r_i (t) : i \in \hV, t \in \hT \}$ that satisfy \cref{eq:rate constraint,eq:power constraint,eq:energy demand constraint}. 
\end{definition}

\noindent
Because \cref{eq:rate constraint,eq:power constraint,eq:energy demand constraint} are affine, verifying the feasibility of an EV charging instance is a linear programming (LP) problem for which many efficient algorithms exist.

\begin{table}[htbp]
	\normalsize
	\centering
	\rowcolors{2}{white}{gray!25}
	\caption{notations}
	\label{table:notations 1}
	\begin{tabularx}{0.475\textwidth}{cX}
		\hline
		\textbf{notation} & \multicolumn{1}{c}{\textbf{description}} \\
		\hline
		\hline
		$\hV$ & set of EVs in the system  \\
		$\hV_t$& set of EVs remain in the charging station at time $t$ \\
		$\hT$ & set of the system's discrete-time  \\
		$a_i$ & arrival time of EV $i$ \\
		$d_i$ & departure time of EV $i$ \\ 
		$e_i$ & energy demand of EV $i\in\hV$ \\
		$e_i(t)$ & remaining energy demand of EV $i$ at time $t\in\hT$ \\
		$r_i (t) $ & charging rate of or energy transmitted into EV $i$ at time $t$ \\
		$P (t)$ & power limit of the charging station at time $t$ \\ 
		${\hI}$ & EV charging problem instance\\
		\hline
	\end{tabularx}
\end{table}

\subsection{Online Scheduling}
\label{SS:Online Scheduling} 

In practice, information of the energy demand and departure time of an EV is only available after its arrival\footnote{In ACN, information of the energy demand and departure time is gathered from user inputs upon arrival.}. Consequently, the charging station need to employ an online algorithm to determine the current charging rate of EV $i$ at time $t$, $r_i(t)$, given information only up to the current time by mapping it into

\begin{equation}
	\label{eq:current instance}
	\begin{aligned}
		& & \hI_t & = \{a_i, d_i, e_{i}(\tau), \bar{r}_{i} ; P(\tau)\}_{i\in \hV_t, \tau\leq t} \\
		\text{where} & & e_{i}(\tau) & = e_{i} - \sum_{j = 0}^{\tau-1} r_i(j)
	\end{aligned}
\end{equation}

\begin{definition}[online algorithm]
	\label{definition:online algrithm}
	An online algorithm is a sequence of functions, $\hA = \{\hA_t\}$, where each function, $\hA_t : {\hI_t} \rightarrow {r(t)}$, maps the information up to the current time, $\hI_t$, to the current charging rates, $r (t)= \{r_i(t)\}_{i\in \hV_t }$. 
\end{definition}

\begin{definition}[feasibility of an algorithm]
	\label{definition:feasibility of an algorithm}
	An (online) algorithm, $\hA$, is (online) feasible on instance $\hI$ if it gives charging rates that satisfy constraints defined in \cref{eq:rate constraint,eq:power constraint,eq:energy demand constraint}.
 \end{definition}

For an online algorithm to be feasible, it must be online feasible for all offline feasible instances. However, In general, there does not exist an online algorithm that is feasible on all offline feasible instances. In this paper, we investigate the cases in which online feasibility can be guaranteed with additional assumptions.

\subsection{Smoothed Least-Laxity-First Algorithm}
\label{SS:Smoothed Least-Laxity-First Algorithm}

In this section, we introduce our proposed online algorithm, the smoothed Least-Laxity-First (sLLF), which is an improvement from the classic least-laxity-first (LLF) algorithm~\cite{Mok1983FundamentalEnvironment}. We can see laxity as a measure for the flexibility (or urgency) in the charging process of an EV.

\begin{definition}[laxity]
	\label{definition:laxity}
	   The laxity of an EV $i \in \hV$ at time $t\in\hT$ is defined as the remaining time of the vehicle in the charging station minus the minimum remaining time needed to be fully charged\footnote{ For convenience, laxity is defined on the whole temporal domain $\hT$.} \ie:

	\begin{equation*}
		\ell_i(t) = 
		\begin{cases}
			[d_i - t ]^+ - \frac{ e_i(t)}{\br_i}, & t \geq a_i \\ 
			+ \infty, & t< a_i
		\end{cases}
	\end{equation*}	   

	\noindent
	where $[\cdot]^+$ denotes as the projection onto the set $\mathcal{R}_+$ of non-negative real numbers. 
\end{definition}

\noindent
Notice that for $t<d_i$,

\begin{align*}
	\ell_i(t+1) & = [d_i - (t+1) ]^+ - \frac{ e_i(t+1)}{\br_i} \\ 
	& = \ell_i(t)-1-\frac{r_i(t)}{\br_i}
\end{align*}	

\noindent
From \cref{definition:feasibility of an algorithm}, we can understand that for an algorithm to be feasible, it needs to satisfy \cref{eq:rate constraint,eq:power constraint,eq:energy demand constraint}. \cref{eq:rate constraint} can be rewrite as $0 \leq \frac{r_i(t)}{\br_i} \leq 1$ which suggests that laxity of EV $i$ is monotonically decreasing at $a_i \leq t<d_i$. Then, \cref{eq:energy demand constraint} implies that, for $t \geq d_i$,  $e_i(t)=0,\ i \in \hV$ which means that for all $t \geq d_i$ 

\begin{align*}
		\ell_i(t ) & = \left[d_i - t) \right]^+ - \frac{ e_i(t )}{\br_i} \\
		& = 0
\end{align*}  

\noindent
Therefore, these feasibility condition implies the following proposition:

\begin{proposition}
	\label{proposition:feasibility condition}
	The algorithm $\hA$ is feasible on an instance $\hI$ if and only if $\hA$ gives charging rates that result in non-negative laxity for all EVs, \ie: 
	
	\begin{equation*}
		\ell_i (t) \geq 0, \; \forall i \in \hV, \; t \in \hT  
	\end{equation*}

\end{proposition}

\cref{proposition:feasibility condition} suggests that the smallest laxity among all EVs can serve as a measure of the distance from infeasibility. A naive approach---referred to as the least laxity first (LLF) algorithm~\cite{Mok1983FundamentalEnvironment}---is to charge EVs starting from those with the least laxity. However, the LLF algorithm may compromise the feasibility of certain offline feasible instances (see \cref{S:Performance Evaluation and Comparison}). The LLF algorithm also cause excessive preemption and oscillations in the charging rate\footnote{\label{FN:LLF_ossi}For example, consider a system of two EVs, where $\ell_1(0) = 1.25,\ \ell_2(0) = 0.75$ and $\br_1 = \br_2 = P(t) =1, t \in \hT$. EV $1$ and EV $2$ will be charged according to $(\rr_1(0), \rr_2(0)) = ( 0, 1)$, $(\ell_1(1), \ell_2(1)) = ( 0.25 ,0.75)$; $(\rr_1(1), \rr_2(1)) = (1, 0 )$, and so on. In this example, both EV switches in-between charging and not charging.}, which may reduce the lifetime of certain batteries (\eg Li-ion)~\cite{7809844}. To eliminate these drawbacks, we propose an alternative approach by maximizing the minimum laxity among all EVs with the charging rate at time $T$, $r(T)=\{r_i(T): \ \forall i \in \hV\}$, as the design variable in order to maximize the feasibility margin, $\max_{r(T)} \min_{ i \in \hV } \ell_i(T)$. 

However, because the solution to the above maximization problem may not be unique, we considered the following problem to produce a unique solution: 

\begin{equation}
	\label{eq:offline_laxity}
	\begin{aligned}
		r(T) =\  & \argmax_{r(T)} & & \sum_{i  \in \hV}  \br_i f( \ell_i(T) ) \\
		& \text{such that} & & \eqref{eq:rate constraint}, \  \eqref{eq:power constraint}, \text{ and } \eqref{eq:energy demand constraint}
	\end{aligned}
\end{equation}

\noindent
where $f$ is twice continuously differentiable, strictly concave, and monotonically increasing. Here, if an instance $\hI$ is offline feasible, then there exists certain charging rates that achieve $\ell_i(T) = 0, \forall i \in \hV$, which yields $\sum_{i  \in \hV}  f( \ell_i(T) ) = \sum_{i  \in \hV}  f(0) $. Since the laxity is monotonically decreasing at any $t \in \hT$, such charging rates also satisfy \cref{proposition:feasibility condition} which implies that \cref{eq:offline_laxity} is feasible on instance $\hI$, \ie:

\begin{corollary}
	\label{corollary:feasibility of optimization problem} 
	\cref{eq:offline_laxity} is feasible for any offline feasible instance.
\end{corollary}

To obtain the solution of \cref{eq:offline_laxity} without information of incoming EVs, we approximate \eqref{eq:offline_laxity} with the following online algorithm: at each time $t \in \hT$, given $\ell_i(t), i \in \hV$, compute\footnote{For more complex form of power limits, in \cref{eq:offline_laxity,eq:online_laxity}, the power constraints in \cref{eq:power constraint} can be replaced by $A r(t) \leq_{e.w.} P(t)$, for element-wise inequality and positive matrix $A$. \cref{corollary:feasibility of optimization problem} also holds for $A r(t) \leq_{e.w.} P(t)$.}

\begin{equation}
	\label{eq:online_laxity}
	\begin{aligned}
		& \max_{r(t)} & & \sum_{i  \in \hV}  \br_i f( \ell_i(t+1) ) \\
		& \text{such that} & & \eqref{eq:rate constraint} \text{ and } \eqref{eq:power constraint} \\
		& & & r_i (t)  \leq e_i (t),  \ i \in \hV_t
	\end{aligned}
\end{equation}

\noindent
\cref{eq:online_laxity} also maximizes the minimum laxity at time $t+1$, $\min_{ i \in \hV_{t} } \ell_i(t+1)$, and thus maximizes the feasibility margin at time $t$\footnote{The solution of \cref{eq:online_laxity} is also unique.}.

To solve \cref{eq:online_laxity}, we first need to look at the Karush-Kuhn-Tucker (KKT) conditions of the problem:

\begin{align}
	\label{eq:kkt1}
	& r_i (t) \geq 0 &\ i \in \hV_t \\
	\label{eq:kkt2}
	& r_i (t)  \leq \min( e_i (t),  \br_i ) &\ i \in \hV_t \\
	\label{eq:kkt3}
	& \sum_{i \in \hV_t} r_i (t) \leq P(t) &\ i \in \hV_t \\
	\label{eq:kkt4}
	&f'( \ell_i(t+1) ) + \bar \lambda_i - \underline \lambda_i + v = 0 & \ i \in \hV_t \\
	\label{eq:kkt5}
	&\underline \lambda_i \geq 0, \;\;\bar \lambda_i\geq 0& \ i \in \hV_t\\
	\label{eq:kkt6}
	&\underline \lambda_i r_i (t)  = 0 , \;\; \bar \lambda_i \{ r_i (t) - \min( e_i (t),  \br_i ) \}= 0 & \ i \in \hV_t
\end{align}

\noindent
where $\underline \lambda_i,\ \bar \lambda_i, \text{ and } v$ are the dual variables for constraints \eqref{eq:kkt1}, \eqref{eq:kkt2}, and \eqref{eq:kkt3} respectively. Consider three mutually exclusive cases: 

\begin{itemize}

	\item $r_i (t) = 0$ that leads to $\bar \lambda_i = 0$ and

		\begin{equation}
			\label{eq:kkt_c1}
			\begin{aligned}
				\frac{r_i(t)}{\br_i}  & = f'^{-1} ( - v ) -  \ell_i(t) + 1 - \underline \lambda_i \\
				& \leq f'^{-1} ( - v ) -  \ell_i(t) + 1
			\end{aligned}
		\end{equation}
		
	\item $r_i (t) \in \{ 0 ,\min( e_i (t),  \br_i ) \}$ which implies $\bar \lambda_i = \underline \lambda_i = 0$; obtained from \eqref{eq:kkt6} (complementary slackness). Substituting $\bar \lambda_i = \underline \lambda_i = 0$ into \eqref{eq:kkt4}, we obtain 

		\begin{equation}
			\label{eq:kkt_c2}
			\begin{aligned}
				\frac{r_i(t)}{\br_i} & = f'^{-1} ( - v ) - \ell_i(t) + 1 \ \ 
			\end{aligned}
		\end{equation}

	\item $r_i (t) = \min( e_i (t),  \br_i ) $ that leads to $\underline \lambda_i = 0$ and
		
		\begin{equation}
			\label{eq:kkt_c3}
			\begin{aligned}
				\frac{r_i(t)}{\br_i}  & = f'^{-1} ( - v ) -  \ell_i(t) + 1 + \hat \lambda_i \\
				& \geq f'^{-1} ( - v ) -  \ell_i(t) + 1
			\end{aligned}
		\end{equation}

\end{itemize}

\noindent
Here, the inverse of $f'$ exists because $f'$ is strictly concave, strictly increasing, and twice continuously differentiable. 

To simplify the notation, define an variable $L(t) = f'^{-1} ( - v )$. Then the following can be obtained:

\begin{proposition}
	
	\label{proposition:valley_filling}
	With $f$ strictly concave, strictly increasing, and twice continuously differentiable, a solution to the optimization problem \eqref{eq:online_laxity} can be obtained by combining \eqref{eq:kkt_c1}, \eqref{eq:kkt_c2}, and \eqref{eq:kkt_c3}:
	
	\begin{equation}
		\label{eq:threshold}
		r_i^*(t)  = [ \br_i (L(t) -  \ell_i(t) + 1 ) ]_0^{\min ( \br_i, e_i(t) ) }, \ \ i \in \hV_t
	\end{equation}
	
	\noindent
	where $[x]^b_a$ denotes the projection of the scalar $x$ on interval $[a,b]$ and $r_i^*(t)$ is  the resulted $r_i(t)$. The solution is then attained at the boundary
	
	\begin{equation}
		\label{eq:threshold2}
		\begin{aligned}
			\sum_{ i \in \hV_t} r^*_i (t) & = \sum_{ i \in \hV_t} [ \br_i ( L(t) - \ell_i(t) + 1 )   ]^{\min ( \br_i , e_i(t) )}_0 \\
			& = \min \left( P(t) , \sum_{i \in \hV_t} \min ( \br_i , e_i(t) ) \right) 
		\end{aligned}
	\end{equation}
	
\end{proposition}

For EV $i \in \hV_t$ with $\br_i \leq e_i(t)$, we obtain charging rate from \cref{eq:threshold} that makes $\ell_i(t+1) = [ L(t) ]_{\ell_i(t) - 1}^{\ell_i(t)}$. Hence, $L(t)$ can be considered as a threshold of $\ell_i(t+1)$. Since $r^*_i(t)$ in \cref{eq:threshold} is an increasing function of $L(t)$, a binary search can be employed to find the threshold $L(t)$ in \cref{eq:threshold2}. Given $L(t)$, the charging rates $r^*_i(t), i \in \hV_t$ is then determined using \cref{eq:threshold}. This procedure is a scalable algorithm that we formally state in \cref{alg:proposed}, and name it as the \textit{smoothed least-laxity-first (sLLF)} algorithm. 

\begin{algorithm}
	\caption{smoothed least-laxity-first (sLLF) }
	\label{alg:proposed}
	\begin{algorithmic}[htbp]
		\FOR{$t \in \hT$}
			\STATE update set of EVs, $\hV_t$, and laxity, $\ell_i (t) \text{ for } i \in \hV_t$
			\STATE obtain $L(t)$ that solves \cref{eq:threshold2} using bisection
			\STATE charge according to rates $r^*_i(t)$ in \cref{eq:threshold}
		\ENDFOR
	\end{algorithmic}
\end{algorithm}

We found that the computational complexity of this sLLF algorithm at each time $t$ is $O(|\hV_t|+\log(1/\delta))$, where $\delta$ is the level of tolerable error. We need $O(|\hV_t|)$ operations to update the laxity of vehicles, and $O(\log(1/\delta)$ operations for binary search for $L(t)$. We also note that the sLLF algorithm possesses the following properties that will be useful for analyzing the feasibility condition:

\begin{enumerate}

	\item Persistence

	\begin{lemma}
		\label{lemma:flip_conditions} 
		Under the sLLF algorithm, if there exist two EVs $i, j \in \hV$ such that 
	
		\begin{equation}
			\label{eq:mn_condition}
			\begin{aligned}
				\ell_i(t) & \leq \ell_j(t),  \\
				\ell_i(t+1) & > \ell_j(t+1)
			\end{aligned}
		\end{equation}
		
		\noindent
		then either one of the following holds: 
	
		\begin{equation}
		\label{eq:flip_cases}
			\begin{cases}
				 t \geq d_i \ \& \  r_i(t) = 0,\\
				 t < d_i\ \& \  t < d_j \ \& \ e_j(t+1) = 0\ \&\ r_i(t)  \neq 0
			\end{cases}
		\end{equation}
	
		\begin{proof}[\textbf{Proof}]
	
			\noindent
			\cref{definition:laxity} satisfies the following relation:
	
			\begin{equation}
				\label{eq:z_bound}
				\begin{aligned}
					\ell_i(t) - 1 & \leq \ell_i(t+1), & i\in\hV \\
					& \leq \ell_i(t)
				\end{aligned}
			\end{equation}
			
			\begin{itemize}
				
				\item In the case  $r_i(t) = 0$:
		 
				\begin{equation}
					\label{eq:z_updates}
					\ell_i(t+1) = 
					\begin{cases} 
						\ell_i(t) -1, & t < d_i \\ 
						\ell_i(t), & t \geq d_i
					\end{cases}
				\end{equation}
				
				\noindent
				Suppose that $t < d_i$, combining the first condition in \eqref{eq:mn_condition} and \eqref{eq:z_bound} gives
		
				\begin{align*}
					\ell_i(t+1) & = \ell_i(t) -1 \\
					& \leq \ell_j(t) -1 \\
					& \leq \ell_j (t+1) \\
					\intertext{thus}
					\ell_i(t+1) & \leq \ell_j (t+1) 
				\end{align*} 
				
				\noindent
				which contradicts the second condition in \eqref{eq:mn_condition}. Therefore, $t \geq d_i$ and it follows the first case in \eqref{eq:flip_cases}. 
				
				\item In the case $r_i(t) \neq 0$, it implies $t < d_i$ then:
				
				\begin{itemize}
		
					\item if $ t<d_j$, \eqref{eq:mn_condition} jointly implies
		
					\begin{equation}
						\label{eq:two_gamma_cond}
						\frac{ r_j(t) }{ \br_j(t) } <  \frac{ r_i(t) }{ \br_i(t) }
					\end{equation}
		
					\noindent
					Under the sLLF algorithm, \eqref{eq:two_gamma_cond} happens only when $e_j(t) =  r_j(t)$, which leads to $e_j(t+1) = 0$. 
		
					\item if $ t \geq d_j$, then 
					
					\begin{align*}
						\ell_j(t+1) & = \ell_j(t) \\
						& \geq \ell_i(t) \\
						& \geq \ell_i (t+1) \\
						\intertext{thus}
					    \ell_i(t+1) & \leq \ell_j (t+1) 
					\end{align*}
					
					\noindent
					which contradicts the second condition in \eqref{eq:mn_condition}. Therefore, it follows the second case in \eqref{eq:flip_cases}.
	
				\end{itemize}
	
			\end{itemize}
	
		\end{proof}
	
	\end{lemma}
	
	\item Fairness

	\noindent
	From Lemma \ref{lemma:flip_conditions}, the solution of the optimization problem \eqref{eq:online_laxity} does not depend of the specific choice of the value function $f$ as long as $f$ is concave, strictly increasing, and has a derivative whose inverse function is well-defined. Without loss of generality, we consider $f( x ) = \log ( x)$. Since non-negative weighted sum and composition with an affine mapping preserve concavity, $C( r(t) ) = \sum_{i \in \hV_t} \br_i f( l_i (t) -1 + r_i(t) / \br_i  )$ is a concave function of $r(t) = [r_1(t) ,r_2(t), \cdots , r_{|\hV_t|}(t)]^T$. Let $\hat r(t) \neq r(t)$ be any rates that satisfy constraints in \cref{eq:rate constraint,eq:power constraint}, where $\hat \ell_i(t) , \ i \in \hV_i$ be the resultant laxity, then, from the first-order-condition of concave functions, $C( \hat r(t) )  - C( r(t) ) + ( r(t)  - \hat r(t) )^T \nabla C( r(t) )  \geq 0 $. Since $r(t)$ is the optimal solution, then
	
	\begin{align*}
		0 & \leq C( r(t) ) -C( \hat r(t) )  \\
		& \leq ( r(t)  - \hat r(t) )^T \nabla C( r(t) ) \\
		& \leq \sum_{ i \in \hV_t } \br_i \frac{ \ell_i(t)  - \hat \ell_i (t) } {  \ell_i(t+1)}
	\end{align*}
	
	\noindent
	where the derivative above is taken over $r(t)$ for $f( \ell_i(t+1) ) = \log ( \ell_i(t+1) )$. On the other hand, if $\hat \ell_i(t+1) > \ell_i(t+1)$ for some EV $i \in \hV_t$, then $r_i(t) < \hat r_i(t) \leq \br_i$. This can only happen when $\ell_i(t+1) \geq  L(t)$ or $\ell_i(t+1) >  L(t)$. As EV $i$ in $\hat r_i(t)$ receives more energy than that in $r_i(t)$, there exists an EV $j$ that receives less energy in $\hat r_j(t)$. Any EV $j$ that receives non-zero energy satisfies
	
	\begin{align*}
		\hat \ell_j(t+1) & \leq L(t+1) = \ell_i(t+1) \\
		\hat \ell_j(t+1) & < \ell_j(t+1)
	\end{align*}
	
	\noindent
	These can be summarize as the following corollary
	
	\begin{corollary}
		\label{corollary:fairness}
		Given the past charging rate $r^{t-1}$, the sLLF algorithm finds a current charging rate $r(t)$ that is both proportionally fair and max-min fair to one-step-ahead laxity. In other words, let $\ell_i(t+1)$ be the one-step-ahead laxity under the sLLF algorithm and $\hat \ell_i(t+1)$ be another laxity produced by a charging rate satisfying the constraints in \cref{eq:online_laxity}, the following two conditions hold:

		\begin{itemize}
	
			\item weighted proportional fairness:
	
			\begin{equation}
				\label{eq:rate constraintop_fair}
				\sum_{ i \in \hV_t } \br_i \frac{\hat \ell_i(t+1) - \ell_i(t+1)  } { \ell_i(t+1) }  \leq 0 
			\end{equation}
			
			\item max-min fairness:  if $\hat \ell_i(t+1) > \ell_i(t+1)$ for some EV $i \in V_t$, then there exits EV $j \in V_t$ such that
	 
			\begin{equation}
				\label{eq:minmax_fair}
				\begin{aligned}
					\hat \ell_j(t+1) &\leq \ell_i(t+1) \\
					\hat \ell_j(t+1) & < \ell_j(t+1)
				\end{aligned}
			\end{equation}
	
		\end{itemize}
	
	\end{corollary} 

\end{enumerate} 

\section{Performance Analysis}
\label{S:Assessment Procedure}

To evaluate our proposed sLLF algorithm, we will compare its performance with several common scheduling algorithms. We will also assess its feasibility condition utilizing the resource augmentation framework. In this section, we present the theoretical background of the resource augmentations (\cref{SS:Augmentation Theory}) as well as the experimental setup (\cref{SS:Experimental Setup}) to evaluate the sLLF algorithm and compare it to several common scheduling algorithms.

\subsection{Resource Augmentation Framework}
\label{SS:Augmentation Theory}

There are two extreme cases in which online algorithms can be feasible for any offline feasible instances: $\br_i \rightarrow \infty$ $\forall i \in \hV$ and $P(t) \rightarrow \infty$. In the first case, $\br_i \rightarrow \infty \text{ } \forall i \in \hV \equiv P(t) \leq min_{i \in \hV_t} \br_i \text{ } \forall t \in \hT$, the charging problem is identical to the single processor preemptive scheduling problem where the processing capacity is time-variant. Here, the earliest-deadline-first (EDF) algorithm is feasible for any offline feasible instances~\cite{Stankovic:1998:DSR:552538}. In the second case, $P(t) \rightarrow \infty \equiv P(t) \geq \sum_{ i \in V_t } \br_i(t) \text{ } \forall t\in\hT$, the sLLF algorithm is feasible for any offline feasible instances. Beyond these two extreme cases, no online algorithm can be feasible on all offline feasible instances~\cite{6670091}. The difficulty in finding feasible online algorithms motivates us to perform a quantitative measurement in evaluating the likelihood of an algorithm become feasible.

From the two cases mentioned above, we can observe that if more resources (\textit{e.g.,} $P(t)$ and $\br_i$) are allowed, an otherwise infeasible problem may become online feasible under an online algorithm. Based on this, we performed a resource augmentation study to characterize the minimum amount of additional resources that will allow an algorithm to produce a feasible solution. Specifically, we analyzed the performance of the sLLF algorithm by adding more (minimum) resources to augment either power supply (power augmentation) or both power supply and peak charging rate (power+rate augmentation). The former augmentation allows more EVs to be charged simultaneously, while the latter allows faster charging. These two augmentation approaches are qualitatively different and provide complementary insights into the behavior of the sLLF algorithm.

Resource augmentation has been studied for processor scheduling in~\cite{kalyanasundaram2000speed,Phillips2002,liu1973scheduling,dertouzos1989multiprocessor,davis2011survey}. The difference is in our setting is that the power limit is time-varying, the maximum rates are heterogeneous, and the power limit may not necessarily be integer multiplication of the maximum rate. Consequently, value of augmentation also depends on the interplay between $P_{\min}$, $P_{\max}$, $\br_{\min}$, and $\br_{\max}$, which complicates the analysis. As mentioned, in this study we considered power and power+rate augmentations. The framework for each of this augmentation will be discussed below:

\subsubsection{\textbf{Power Augmentation}}
\label{SSS:Power Augmentation}

In this augmentation, we allowed online algorithm to utilize $\epsilon$ more power such that

\begin{align*}
	P^{online}(t) & = (1+\epsilon) P(t) \\ 
	\intertext{but} 
	\br_i^{online} & = \br_i
\end{align*}

\noindent
We will call this augmentation as $\epsilon$-power augmentation, where

\begin{definition}{[$\epsilon$-power augmented instance]}~
	\label{definition:pa_inst}
	Given an EV charging instance $\hI = \{a_i, d_i, e_{i}, \bar{r}_{i}; P(t)\}_{i\in\hV, t\in\hT}$, the instance under $\epsilon$-power augmentation is defined as
	
	\begin{equation*}
		\label{eq:pa}
		\{a_i, d_i, e_{i}, \bar{r}_{i};  (1 + \epsilon ) P(t)\}_{i\in\hV, t\in\hT}
	\end{equation*}

\end{definition}

\begin{definition} [$\epsilon$-power feasibility]
	\label{definition:pa_feasibl}
	An online algorithm $\hA$ is $\epsilon$-power feasible if $\hA$ is feasible on the $\epsilon$-power augmented instances $\hI_p(\epsilon)$ generated from any offline feasible instance $\hI$.\footnote{Alternatively, the (minimum) value of $\epsilon$ can also be interpreted as the constraints on instances that are online feasible. That is, given the original resource $P(t), \br_i (t)$, the algorithm is online feasible for any instances $\hI=\{a_i, d_i, e_{i}, \bar{r}_{i}; P(t)/(1+\epsilon)\}_{i\in\hV, t\in\hT}$ that is offline feasible given the reduced resource $P(t)/(1+\epsilon), \br_i (t)$. Large $\epsilon$ restricts possible instances, thus less likely to be online infeasible.}  
\end{definition}

Unfortunately, there is no online algorithm that $\epsilon$-power feasible for any finite $\epsilon > 0$~\cite{Phillips2002}\footnote{It is shown in \cite{Phillips2002} that the LLF algorithm is not $\epsilon$-power feasible for any $\epsilon > 0$ for uniform processors and time-invariant number of processors. Since their setting is a special case of our setting, the same results extend to our setting.}. However, under a mild assumption, the $\epsilon$-power feasibility condition can be obtained for a finite $\epsilon$. Assume that the energy demand of each EV is bounded by $X$ and that the inter-arrival time between consecutive arrivals is greater than $N$, \ie:

\begin{equation}
	\label{eq:pa_assumptions}
	\begin{aligned}
		& e_i \leq X, & & i  \in \hV \\
		& | a_i - a_j | > N , & & i ,j \in \hV
	\end{aligned}
\end{equation}

\noindent
where, the value of $N$ can be controlled by choosing appropriate sojourn time for a sampling intervals (the shorter the sojourn time, the smaller the value of $N$) and the value of $X$ can be obtain from maximum battery capacity for common EVs. Then, it can be proven (see \cref{apdxS:proof of th 1}) that we can characterize the relation between $N$ and the sufficient amount of resource augmentation $\epsilon$ as follows:

\begin{theorem}
	\label{thm:power_aug_LLF}
	If both conditions in \eqref{eq:pa_assumptions} hold, then the sLLF algorithm is $\epsilon$-power feasible with

	\begin{equation*}
		\epsilon = \frac{ P_{\max} }{ P_{\min} }  \left( \log_\varphi \left(   \frac{ \sqrt{ 5 } X } { N P_{\max}  } + \frac{1}{2} \right) + 2  \right) - 1 
	\end{equation*}
	
	\noindent
	where $\varphi = \frac{1+\sqrt(5)}{2} \approx 1.61803$ is the golden ratio. 
\end{theorem}

\noindent 
Now, if the inter-arrival time is equal to $N$ and the power demand is equal to $X$, then the incoming energy demand per unit time is $\frac{X}{N}$. Since the total power supply is $P_{\max}$ per unit time, $N$ should be at least $\frac{X}{P_{max}}$ for offline feasibility which is a mild assumption. With this, we then can apply a special condition to \cref{thm:power_aug_LLF}:

\begin{corollary}
	\label{corollary:power_aug_LLF}
	For constant power limit $P(t) = P, t \in \hT$, and $N \geq \frac{X}{P_{\max}}$, then 
	
	\begin{align*}
	    \epsilon & \leq \frac{ P }{ P }  \left( \log_\varphi \left(   \frac{ \sqrt{ 5 } X } { \left(\frac{X}{P_{\max}}\right) P_{\max}  } + \frac{1}{2} \right) + 2  \right) - 1 \\
	    & = \log_\varphi \left( \sqrt{ 5 } + \frac{1}{2} \right) + 1 \approx 3.091639884 \\
	    \intertext{thus}
	    \epsilon & \lesssim 3.091639884
	\end{align*}
	
	\noindent
	Therefore, the sLLF algorithm is approximately $3$-power feasible.
\end{corollary}

\subsubsection{\textbf{Power+Rate Augmentation}}
\label{SSS:Power and Rate Augmentation}

In this case, online algorithm is allowed to utilize $\epsilon$ more power and higher maximum charging rate such that

\begin{align*}
	P^{online}(t) & = (1+\epsilon) P(t) \\ 
	\intertext{and } 
	\br_i^{online} & = (1+\epsilon) \br_i
\end{align*}

\noindent
We will call this augmentation as $\epsilon$-power+rate augmentation, where

\begin{definition}{[$\epsilon$-power+rate augmented instance]}~
	\label{definition:pra_inst}
	Given an EV charging instance $\hI = \{a_i, d_i, e_{i}, \bar{r}_{i}; P(t)\}_{i\in\hV, t\in\hT}$, we define the $\epsilon$-power+rate augmented instance as

	\begin{equation*}
		\label{eq:pra}
		\{a_i, d_i, e_{i},  (1 + \epsilon ) \bar{r}_{i};  (1 + \epsilon ) P(t)\}_{i\in\hV, t\in\hT}
	\end{equation*}

\end{definition}

\begin{definition} [$\epsilon$-power+rate feasibility]
	\label{definition:pra_feasibl}
	An online algorithm $\hA$ is $\epsilon$-power+rate feasible if $\hA$ is feasible on the $\epsilon$-power+rate augmented instances $\hI_{pr}(\epsilon)$ generated from any offline feasible instance $\hI$.   
\end{definition}

However, unlike the case of power augmentation, the sLLF algorithm is $\epsilon$-power+rate feasible for a finite value of $\epsilon > 0$ without any assumptions of the arrival patterns (see \cref{apdxS:proof of th 2}).

\begin{theorem}
	\label{thm:achievable_rate}
	The sLLF algorithm is $\epsilon$-power+rate feasible with
	
	\begin{equation*}
		\epsilon = \max_{i \in \hV} \left( \max_{ \tau_1 , \tau_2 \in [ a_i, d_i ]} \frac{P (\tau_1) }{ P (\tau_2) }  -  \max_{\tau \in [a_i, d_i]} \frac{  \br_i }{ P(\tau) } \right)
	\end{equation*}
 
\end{theorem}

\subsection{Experimental Setup}
\label{SS:Experimental Setup}

We employed trace-based simulation on real EV datasets from the ACN deployment (CAGarage) as well as Google's facilities in Mountain View (Google\_mtv) and Sunnyvale (Google\_svl) to evaluate the performance of our proposed algorithm. The datasets contain a total of 52,362 charging sessions over 
more than 4,000 charging days in 2016 at 104 locations (\cref{table: instance_stat}  provides a summary of the data), an instance consists of one charging day. We compute the minimum power capacity in which each instance is feasible by using an offline LP, \textit{i.e.,} we minimize $P(t)$ subject to \cref{eq:rate constraint,eq:power constraint,eq:energy demand constraint}, which corresponds to the minimum power supply for the instance to be offline feasible. We used this minimum power supply to generate an offline instance and tested if the instance is feasible under an online algorithm. 

\begin{table}[htbp]
	\normalsize
	\centering
	\begin{tabular}{|l|c|c|c|}
		\hline
		\multicolumn{1}{|c|}{datasets} & instances & \begin{tabular}{@{}c@{}} EV sojourn \\ time (min) \end{tabular} & laxity (min) \\
		\hline\hline
		CAGarage & 92 & \begin{tabular}{@{}c@{}} 321 \\ (11, 720) \end{tabular} & \begin{tabular}{@{}c@{}} 231 \\ (0.1, 660) \end{tabular} \\
		\hline
		Google\_mtv & 3793 & \begin{tabular}{@{}c@{}} 149 \\ (10, 720) \end{tabular} & \begin{tabular}{@{}c@{}} 35 \\ (0.001, 694) \end{tabular} \\
		\hline
		Google\_svl & 246 & \begin{tabular}{@{}c@{}} 152 \\ (11, 720) \end{tabular} & \begin{tabular}{@{}c@{}} 38 \\ (0.02, 676) \end{tabular} \\
		\hline
	\end{tabular}
	\caption{Statistics of the EV charging instances that show the average of the sojourn times and laxity in minutes unit; the minimum and maximum values are indicated inside the brackets.}
	\label{table: instance_stat}
\end{table}

We compared the performance of the sLLF algorithm against several common heuristic online scheduling algorithms: earliest-deadline-first (EDF)~\cite{Liu1973SchedulingEnvironment,Stankovic:1998:DSR:552538}, least-laxity-first (LLF)~\cite{Mok1983FundamentalEnvironment}, equal/fair share (ES)~\cite{Kay1988AScheduler}, remaining energy proportional (REP)~\cite{}, and an online linear program (OLP)~\cite{Guo}. The implementation of these algorithms for the current problem can be summarized as follows:

\begin{itemize}
    \label{itm:algorithms summary}

	\item In the EDF algorithm, all EVs in $\hV_t$ are sorted by their deadlines $(d_i)$ in increasing order. The available power at a given time, $P(t)$, is assigned to EVs in this order up to $\min (\br_i , e_i(t))$.
	
	\item In the LLF algorithm, all EVs in $\hV_t$ are sorted by their laxity, $\ell_i(t)$, in increasing order. The available power at a given time, $P(t)$, is assigned to EVs in this order up to $\min (\br_i , e_i(t))$.
	
	\item In the ES algorithm, the available power supply at a given time, $P(t)$, is divided to all connected EVs such that each of them receives the minimum between their fair share and maximum charging rate. The procedure is repeated until either all $P(t)$ has been distributed or there is no more EV can be charged further.
	
	\item In the REP algorithm, the available power supply at a given time, $P(t)$, is divided to EVs in proportion to their remaining energy demand, $e_i(t)$. Here, each EV receives the minimum between their proportional share and maximum charging rate. The procedure is repeated until either all $P(t)$ has been distributed or there is no more EV can be charged further.
	
	\item In the OLP algorithm, the charging rate of EV $i \in \hV$ at a given time, $r_i(t)$, is provided according to the solution of the following LP:
	
	\begin{equation}
		\label{eq:OLP}
		\begin{aligned}
			r_i(t) = \ & \argmin_{r_i(t)} & & \sum_{i \in \hV_t} \sum_{\tau=t}^T \tau r_i(\tau) \\
			& \text{subject to} & & \sum_{\tau=t}^T r_i(\tau) = e_i(t),\  \forall i \in U_t \\
			& & & \sum_{i \in U_\tau} r_i(\tau) \le P(\tau), \  \forall \tau = t, \ldots, T \\
			& & & 0 \le r_i(t) \le \br_i \\
		\end{aligned}
	\end{equation}
	
	Here, the objective function encourages the charging station to charge EVs as early as possible, while the constraints ensure the online LP finds a feasible schedule for all the currently active EVs assuming no future arrivals.  
	
\end{itemize}

\noindent
The performance comparison of the proposed sLLF algorithm against the above algorithms can be seen in \cref{S:Performance Evaluation and Comparison} below.

\section{Performance Evaluation and Comparison}
\label{S:Performance Evaluation and Comparison}

In this section, we will evaluate the performance of the sLLF algorithm and compare it against other algorithms listed in \cref{itm:algorithms summary}. We will first evaluate the success rate of the online algorithms without resource augmentation (\cref{SS:Without Resource Augmentation}) before further analyze their performance with \textbf{1)} power augmentation and \textbf{2)} power+rate augmentation (\cref{SS:With Resource Augmentation}). For this purpose, we define the success rate of an algorithm as the ratio of online feasible instances under the algorithm to all existing instances.

\subsection{Without Resource Augmentation}
\label{SS:Without Resource Augmentation}

\begin{figure}[htbp]
	\centering
	\includegraphics[width=.45\textwidth]{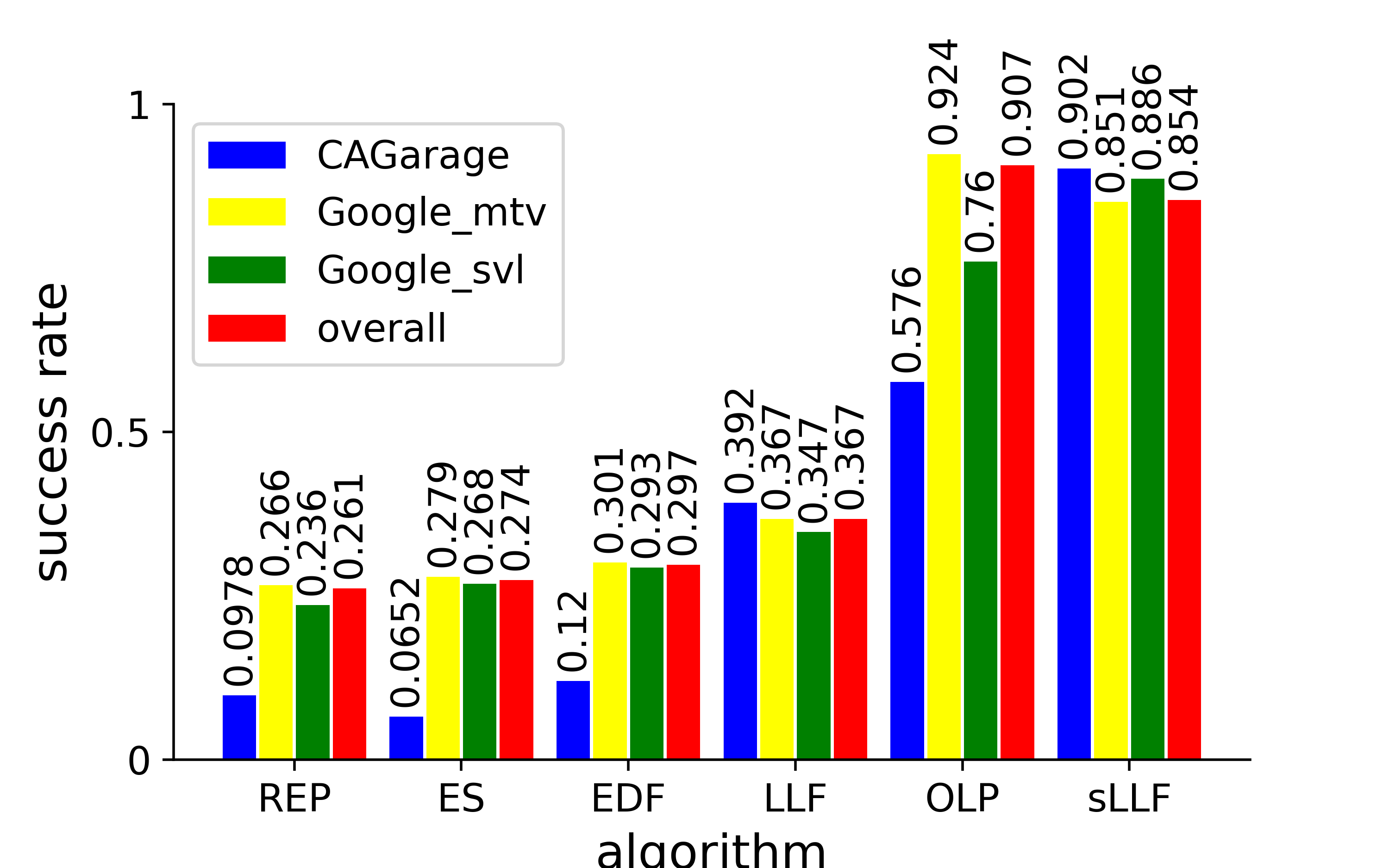}
	\caption{Bar chart showing the success rate of the utilized algorithms in finding a feasible online schedule from different datasets without resource augmentation. The displayed values are rounded to three significant figures.}
	\label{F:SuccessRate_no_Augmentation}
\end{figure}

Comparing the success rate of the sLLF algorithm against different algorithms, summarised in \cref{itm:algorithms summary}, we found that our proposed algorithm achieves a more uniform high success across different datasets (see \cref{F:SuccessRate_no_Augmentation}). From \cref{F:SuccessRate_no_Augmentation} we can also see that the EDF, ES, and REP algorithms perform much worse in terms of finding feasible schedules as expected because these algorithms do not consider the deadline, maximum charging rate, and remaining energy of each EV simultaneously which are necessary to find the feasibility. We can also see that, despite its similarity, the LLF algorithm achieve lower success rate than the sLLF algorithm that suggests the importance of maximizing minimum laxity to eliminate the infeasibility of certain offline feasible instances in the LLF algorithm (see \cref{SS:Smoothed Least-Laxity-First Algorithm}). 

Moreover, although the OLP algorithm achieves a higher success rate in finding a feasible online schedule from Google\_mtv dataset, it requires solving LP problem at every time-step. With the currently available LP solver, the computational complexity for solving LP problem of size $\bf{n}$ will be greater than $O(\bf{n}^2)$ \cite{Jiang2020FasterDM}. Thus, at every given time $t$ the OLP algorithm has computational complexity higher than $O(|\hV_t|^2)$ which is computationally more expensive than the sLLF algorithm that has complexity of $O(|\hV_t|+\log(1/\delta))$. 

\begin{figure}[htbp]
	\centering
	\begin{subfigure}{.24\textwidth}
		\caption{}
		\includegraphics[width=\textwidth]{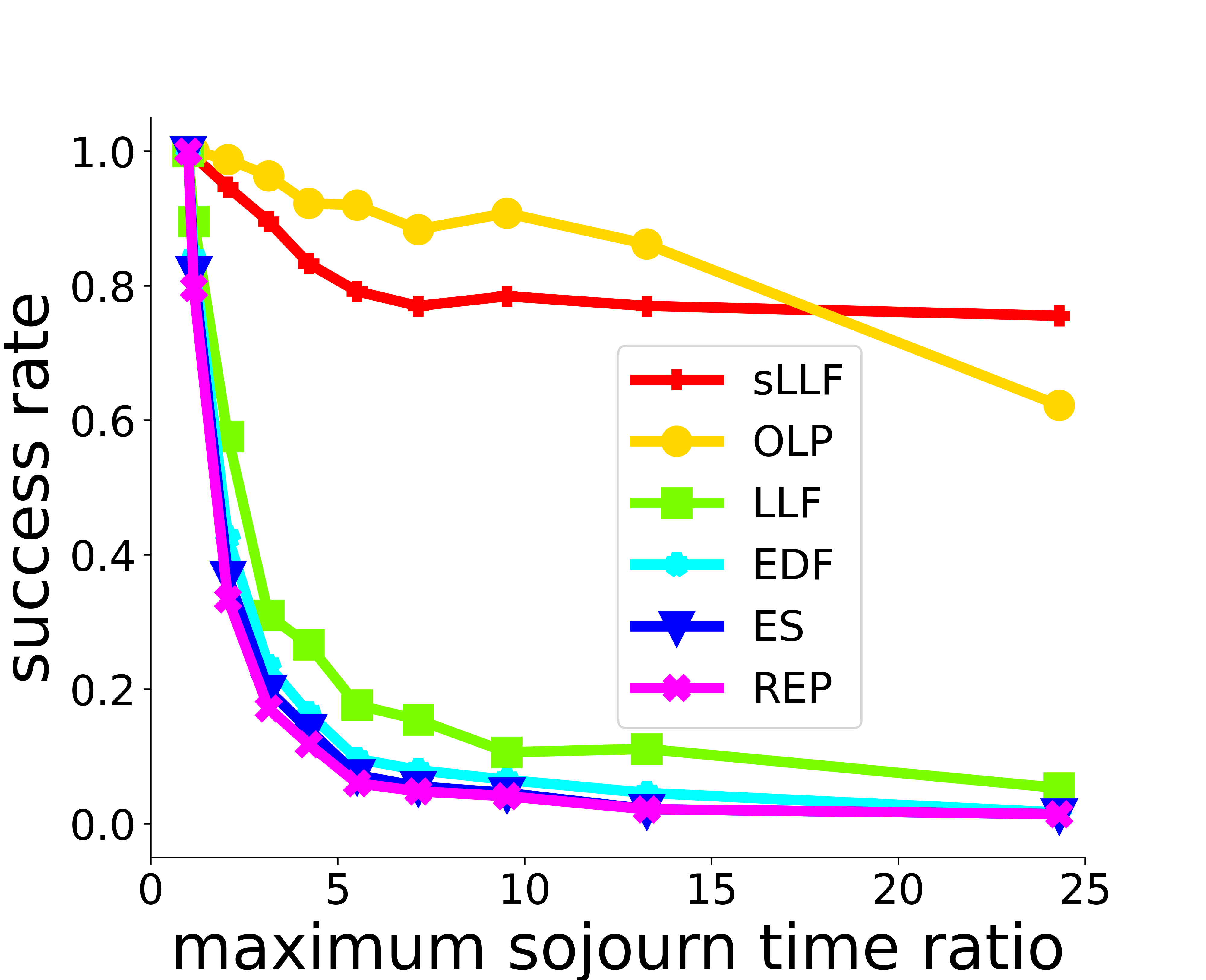}
		\label{SF:SuccessRate_vs_ratio}
	\end{subfigure}
	\begin{subfigure}{.24\textwidth}
		\caption{}
		\includegraphics[width=\textwidth]{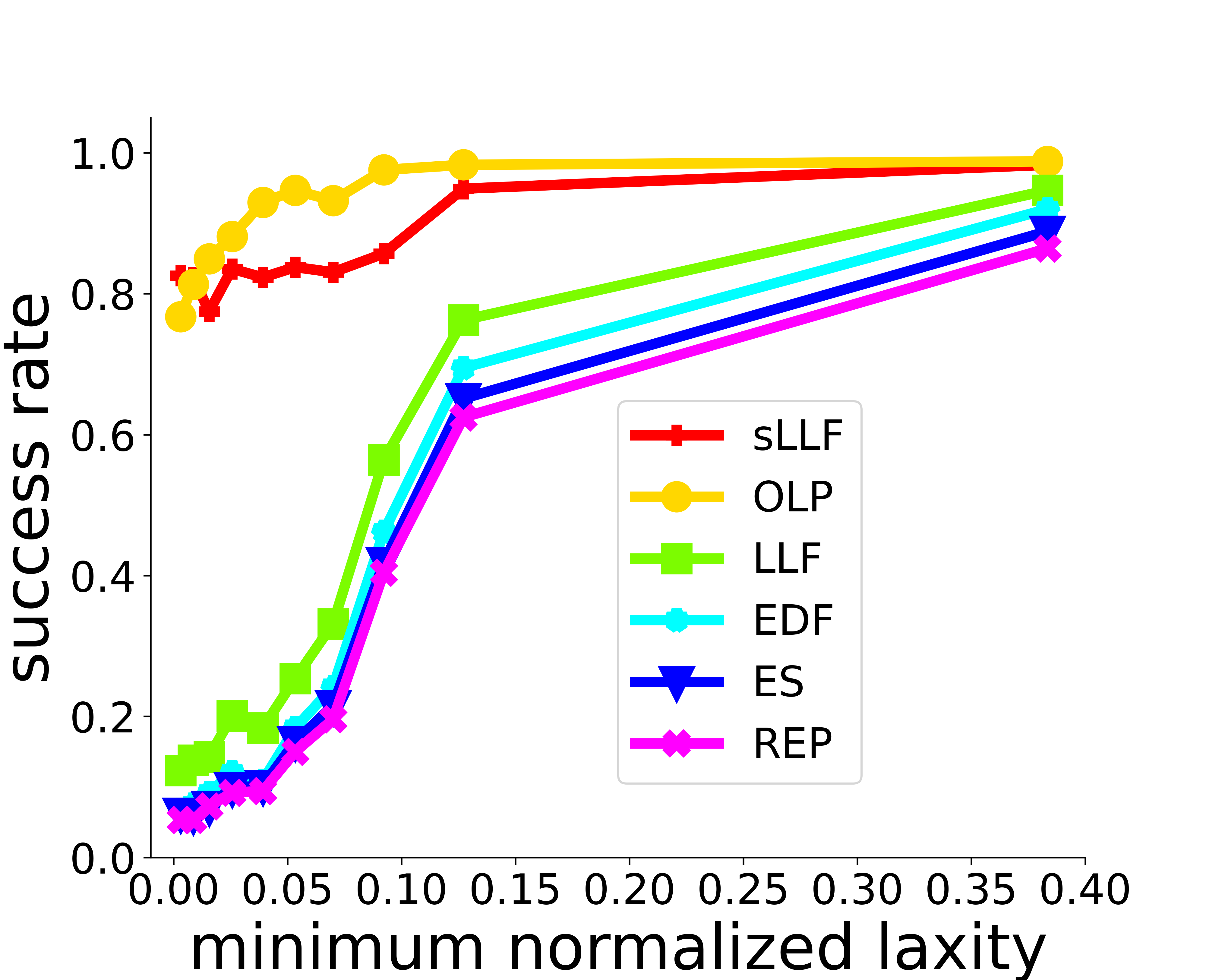}
		\label{SF:SuccessRate_vs_Laxity}
	\end{subfigure}
	\caption{Plots of success rate in finding a feasible online schedule without resource augmentation versus (a) the maximum ratio between EV sojourn times and (b) the minimum normalized laxity.}
	\label{F:SuccessRate_no_Augmentation_vs_RatioAndLaxity}
\end{figure}

Furthermore, we observed that the minimum normalized laxity and the maximum ratio between EV sojourn times have high correlations with the success rate of the algorithms (see \cref{F:SuccessRate_no_Augmentation_vs_RatioAndLaxity}). Here, the maximum ratio between EV sojourn times is defined as the maximum ratio between the longest and shortest EV sojourn times in the instances while the minimum normalized laxity of an EV is defined as the laxity divided by the EV sojourn times, $\ell_i(a_i)/(d_i - a_i)$. To study this, we categorized the dataset we have into different sets of data and the success rate of the algorithms based on these different data categories can be found in \cref{F:SuccessRate_no_Augmentation_vs_RatioAndLaxity}. \cref{SF:SuccessRate_vs_ratio} shows that as the maximum ratio between EV sojourn times increases, all algorithms considered have decreased success rates. This indicates that a large degree of variation in the sojourn time may decrease the performance of online scheduling algorithms. From the \cref{SF:SuccessRate_vs_ratio} we can also see that the sLLF algorithm is least sensitive to the changing of the maximum ratio between EV sojourn times while still maintaining a high success rate. This shows the benefit of the sLLF algorithm against a large variation of EVs' sojourn time that common in real-world applications.

Meanwhile, \cref{SF:SuccessRate_vs_Laxity} shows that higher minimum normalized laxity improves the algorithms' success rate which implies that shorter sojourn time is more desirable to improve the performance of the scheduling algorithms. The result shown in \cref{SF:SuccessRate_vs_Laxity} also indicates that larger laxity gives a greater advantage in the scheduling system which is expected as a less urgent environment is easier to maintain. As we can see in \cref{SF:SuccessRate_vs_Laxity}, the sLLF algorithm has one of the highest success rates for all minimum normalized laxity even when the minimum normalized laxity is small. This shows the benefit of the sLLF algorithm in a high urgency scheduling environment such as in some public charging stations. Additionally, a larger laxity can also be associated with higher resources which leads to the benefit of higher power supply and/or peak charging rate as will be discussed further in \cref{SS:With Resource Augmentation}.

\subsection{With Resource Augmentation}
\label{SS:With Resource Augmentation}

\begin{figure}[htbp]
	\centering
	\begin{subfigure}{.24\textwidth}
		\caption{power augmentation}
		\includegraphics[width=\textwidth]{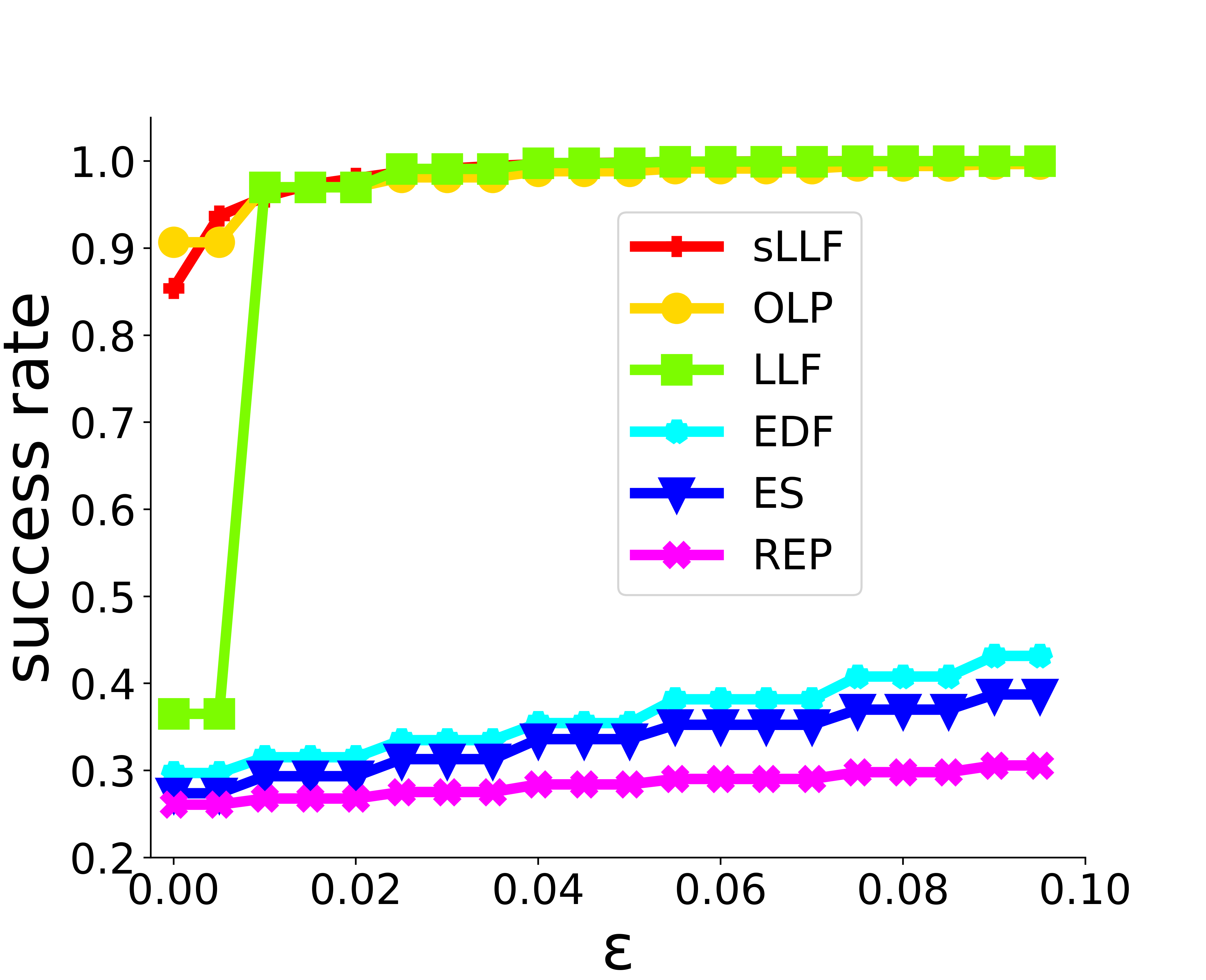}
		\label{SF:SuccessRate_with_PowerAugmentation}
	\end{subfigure}
	\begin{subfigure}{.24\textwidth}
		\caption{power+rate augmentation}
		\includegraphics[width=\textwidth]{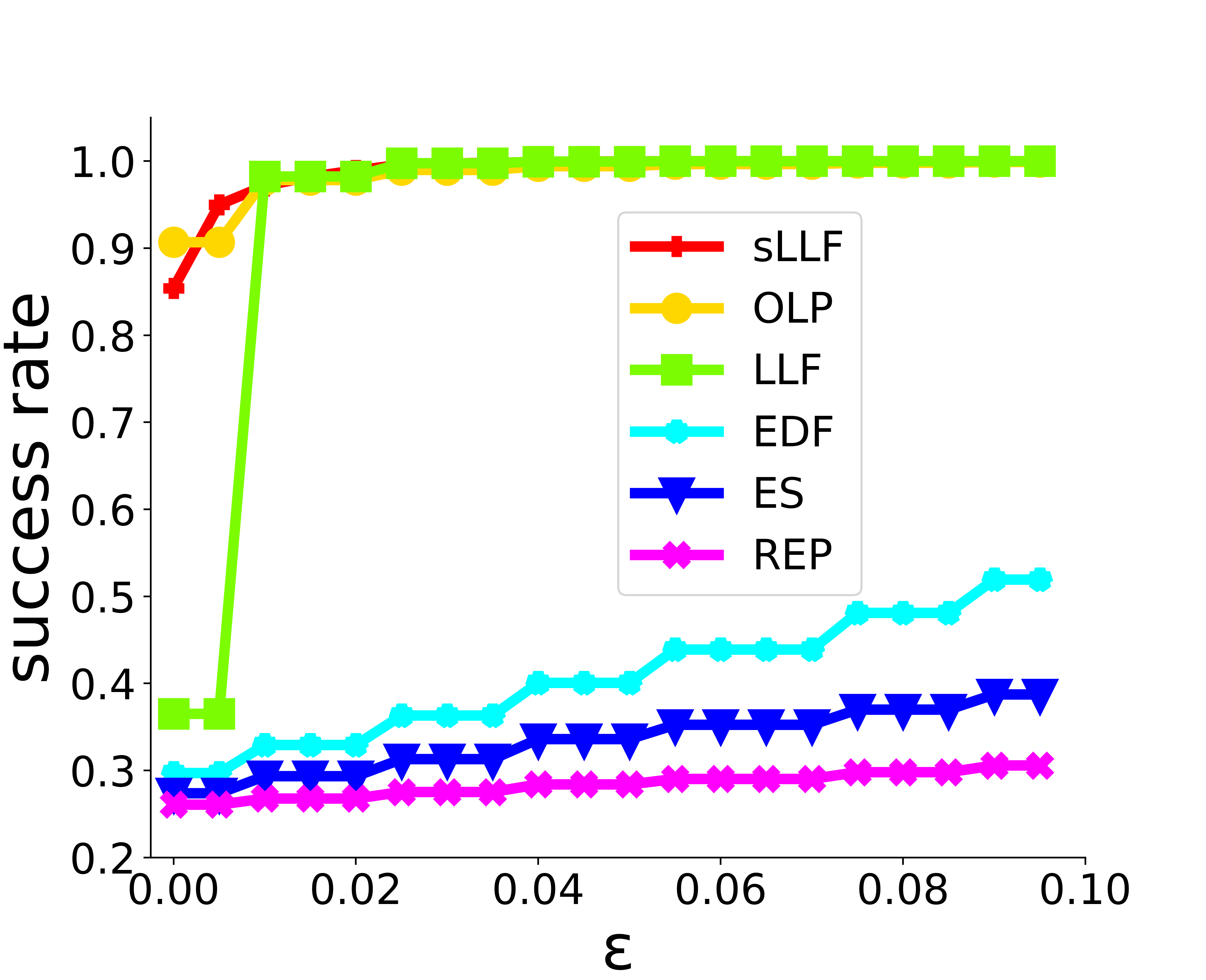}
		\label{SF:SuccessRate_with_Power+RateAugmentation}
	\end{subfigure}
	\caption{Plots of success rate in finding a feasible online schedule with (a) power supply augmentation and (b) power supply along with peak charging rate augmentation.}
	\label{F:SuccessRate_with_Augmentations}
\end{figure}

We analyzed the performance of the online algorithms with resource augmentation in \textbf{a)} power supply and \textbf{b)} both power supply along with peak charging rate to gain further insight into the algorithms' behavior (see \cref{F:SuccessRate_with_Augmentations}). As expected, the success rate of all algorithms increase with more available resources that suggest the benefit of higher power supply and/or peak charging rate in the scheduling system. We can see from \cref{F:SuccessRate_with_Augmentations} that sLLF and OLP algorithms have the highest success rate among other algorithms under the various level of resource augmentation. Although the performance of the sLLF algorithm in the event without resource augmentation is lower than the OLP algorithm, it can achieve a $0.95$ success rate with only a $0.02$ increase in resources.

\begin{table}[htbp]
	\normalsize
	\centering
	\begin{tabular}{|l|c|c|c|c|c|c|}
		\hline
		\multirow{2}{*}{augmentation} & \multicolumn{6}{c|}{$\epsilon$} \\
		\cline{2-7}
		& REP & ES & EDF & LLF & OLP & sLLF \\
		\hline
		\hline
		power supply & 4.61 & 3.65 & 1.39 & 0.07 & 0.28 & 0.07 \\
		\hline
		\begin{tabular}{@{}l@{}} power supply \\ and peak \\ charging rate \end{tabular}  & 4.61 & 3.24 & 0.54 & 0.05 & 0.28 & 0.05 \\
		\hline
	\end{tabular}
	\caption{Minimum resource augmentation to achieve a perfect success rate in finding a feasible online schedule for all instances.}
	\label{table: eps min}
\end{table}

Inspecting further, we listed in \cref{table: eps min} the minimum resource augmentation required for each algorithm to achieve $100\%$ feasibility for all instances. From the table, we can see that, together with the LLF algorithm, our sLLF algorithm has the smallest $\epsilon$ among the algorithms considered. The proposed algorithm can achieve perfect feasibility using only $0.07$ power augmentation which is significantly smaller than the predicted value in \cref{corollary:power_aug_LLF}. Thus, our proposed algorithm has the potential in reducing the infrastructure cost for EV charging facility which will also be beneficial in an application where the resources are limited.

\begin{figure}[htbp]
	\centering
	\begin{subfigure}{.24\textwidth}
		\caption{LLF}
		\includegraphics[width=\textwidth]{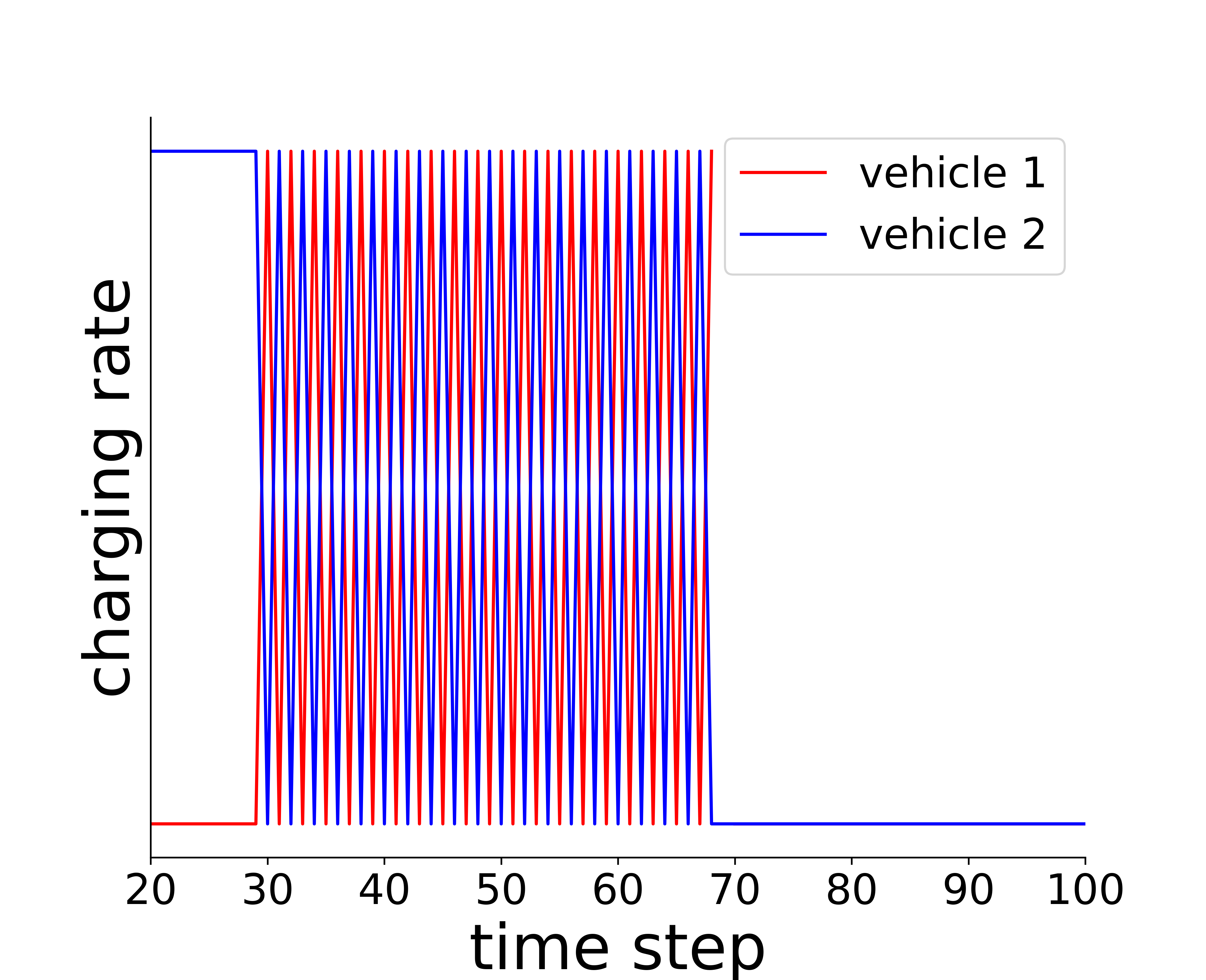}
		\label{SF:oscillation_LLF}
	\end{subfigure}
	\begin{subfigure}{.24\textwidth}
		\caption{sLLF}
		\includegraphics[width=\textwidth]{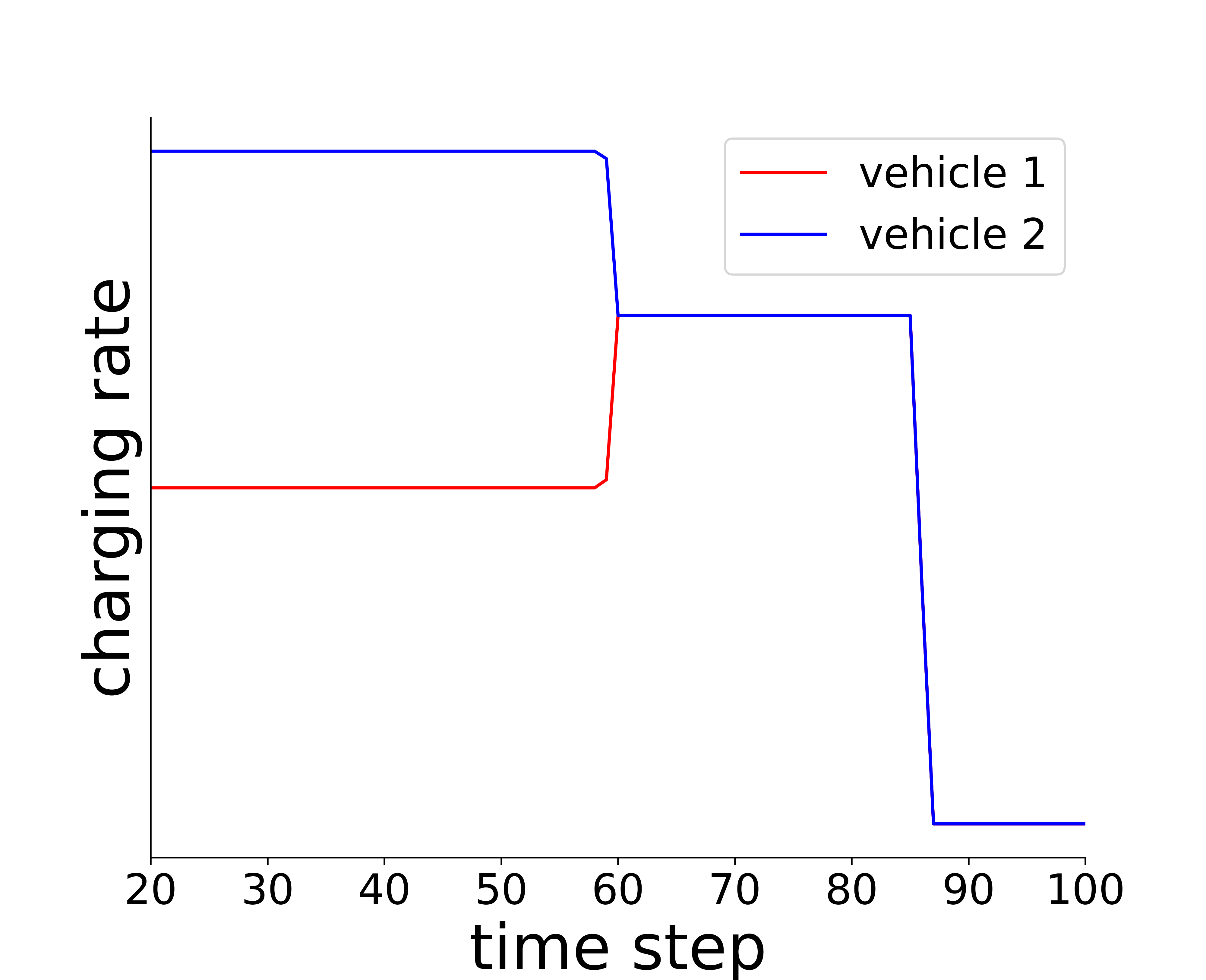}
		\label{SF:oscillation_sLLF}
	\end{subfigure}
	\caption{Charging rate from two vehicles at each time step of a hypothetical case that exaggerate the oscillation behaviour in the LLF algorithm obtained using (a) the LLF and (b) the sLLF algorithms.}
	\label{F:oscillation_comparison}
\end{figure}

Additionally, although the LLF and sLLF algorithms have equal $\epsilon$ feasibility, the sLLF algorithm, as implies by \cref{lemma:flip_conditions}, does not exhibit undesirable oscillations behavior such can be found in the LLF algorithm (\cref{FN:LLF_ossi}). To inspect this property, we simulated a hypothetical case that exaggerate the oscillation behavior in the LLF algorithm. In the simulation we introduced two vehicles with equal maximum charging rates that arrive and will depart at the same time where there isn't other vehicle present at the charging station. The charging rate at each time step of this simulation obtained using the LLF and sLLF algorithms can be seen in \cref{SF:oscillation_LLF,SF:oscillation_sLLF} respectively. The result in \cref{F:oscillation_comparison}, agrees with \cref{lemma:flip_conditions}, shows that the sLLF algorithm eliminates the oscillation behavior that can reduce the lifetime of certain batteries.

\section{Conclusion}
\label{S:Conclusion}

In this work, we formulated EV charging as a feasibility problem that meets all EVs' energy demands before departure under the individual charging rate of each EV and the total power resource constraints. We then proposed an online algorithm, the sLLF algorithm, that decides on the current charging rates based on only the information up to the current time. We characterized and analyzed the performance of the sLLF algorithm analytically and numerically utilizing the resource augmentation framework, where we demonstrated the first application of the framework for heterogeneous processors with a time-varying number.

Our numerical experiments with real-world datasets showed that our algorithm has a significantly higher rate of generating feasible EV charging than several other common EV charging algorithms. We showed that our sLLF algorithm is able to maintain a high success rate and less sensitive to a large variety of EVs' sojourn time that common in a real-world application. The algorithm also shows its benefit in a high urgency scheduling environment such as in some public charging stations. By finding feasible EV charging schedules using only a small augmentation that is also significantly less than the theoretical upper bound, our proposed algorithm (sLLF) can significantly reduce the infrastructural cost for EV charging facilities. Among the algorithms that achieve the highest success rate (\ie the sLLF, LLF, and OLP algorithms), our sLLF algorithm does not exhibit undesirable oscillations such found in the LLF algorithm and computationally cheaper than the OLP algorithm.

\appendix
\crefalias{section}{appsec}
\crefalias{subsection}{appsec}

To provide proofs to the \cref{thm:power_aug_LLF,thm:achievable_rate} we presented in our resource augmentation framework (\cref{SS:Augmentation Theory}), we first, in addition to the notations in \cref{table:notations 1}, introduce some additional notation that will be used in this section, summarized in \cref{table:notations 2}. Here, $A_t=\{ i \in \hV: a_i \leq t\}$ denotes the set of EVs that have arrived by time $t$; $D_t=\{ i \in A_t: d_i \leq t \text{ or } e_i(t) = 0  \}$ denotes the set of EVs that have either departed or finished charging by time $t$; and $U_t= \{ i \in V_t :  e_i(t) > 0\}$ denotes the set of EVs with unfulfilled energy demand at the beginning of time slot $t$. In addition, $A_{[t_1,t_2]}  = \{ i \in \hV: a_i \in [t_1,t_2] \}$ denotes the set of EVs that arrive during time interval $[t_1,t_2]$, $t_1, t_2 \in \hT$. We define $\epsilon$ as the fraction of additional resources to be augmented such that if we augment the power supply $(P)$ then the total power becomes $(1+\epsilon)P$, similarly, for the peak charging rate $(\bar{r})$ augmentation the total rate becomes $(1+\epsilon)\bar{r}$. With this, we also define notations for the total energy supply to EVs in set $\hS\subseteq\hV$ during the interval $[t_1, t_2]$, $\sum_{i \in \hS} \sum_{\tau = t_1}^{t_2} r_i(\tau)$, under instance $\hI$ as $\Psi_{[t_1:t_2]} (\hS; \hI)$. Additionally, the notation for this total energy under an (feasible) offline algorithm and an online algorithm with $\epsilon$ resource augmentation be denoted as $\Psi^*_{[t_1:t_2]} (\hS; \hI)$ and $\Psi^{\epsilon}_{[t_1:t_2]} (\hS; \hI)$ respectively\footnote{In general, we use superscript $^*$ to indicate variables under an (feasible) offline algorithm with original power limit $P(t)$ and maximum charging rates $\bar{r}_i$, where we use superscript $^\epsilon$ to indicate variables under the augmented resources.}.

\begin{table}[htbp]
	\normalsize
	\centering
	\rowcolors{2}{white}{gray!25}
	\caption{additional notations}
	\label{table:notations 2}
	\begin{tabularx}{0.475\textwidth}{cX}
		\hline
		\textbf{notation} & \multicolumn{1}{c}{\textbf{description}} \\
		\hline
		\hline
		$A_t$ & set of EVs have arrived by time $t$ \\
		$A_{[t_1,t_2]}$ & set of EVs arrive during time interval $[t_1,t_2]$ \\
		$D_t$ & set of EVs have either departed or finished by time $t$ \\
		$U_t$ & set of EVs that are still charging at time $t$ \\
		$\epsilon$ & fraction increase in resources \\
		$\Psi_{[t_1:t_2]} (\hS; \hI)$ & total energy supplied to the set of EVs $\hS$ during the interval $[t_1, t_2]$ under instance $\hI$ \\
		$\Psi^*_{[t_1:t_2]} (\hS; \hI)$ & $\Psi_{[t_1:t_2]} (\hS; \hI)$ under an offline algorithm  \\
		$\Psi^\epsilon_{[t_1:t_2]} (\hS; \hI)$ & $\Psi_{[t_1:t_2]} (\hS; \hI)$ under an online algorithm with $\epsilon$ resources augmentation \\
		\hline
	\end{tabularx}
\end{table}

Next, we look into the infeasibility condition of the sLLF algorithm. For a charging instance, ${\hI}=\{a_i, d_i, e_{i}, \bar{r}_{i}; P(t)\}_{i\in\hV, t\in\hT}$, that is not online feasible under the sLLF algorithm, there are times when some EVs has negative laxity; we denote the earliest among such times as $t_-$. Let $\hF = \{ i \in A_{t_-}:  \ell_i (t_- ) < 0 \}$ be the set of EVs arriving at the changing station by time $t_-$ that have negative laxity, $\hS_1= \{ i \in A_{t_-} :  \ell_i ( t_-)  \geq 0\ \&\ d_i  \leq t_- \}$ be the set of EVs with non-negative laxity that depart by time $t_-$, and $\hS_2= \{ i \in A_{t_-} :  \ell_i ( t_- )  \geq 0\ \& \  d_i >  t_- \}$ be the set of EVs with non-negative laxity that remain at the charging station at time $t_-$. Here, $\hF$, $\hS_1$, and $\hS_2$ are mutually exclusive such  $A_{t_-} = \hF \cup \hS_1 \cup \hS_2$. Then:

\begin{lemma}
	\label{lemma:laxre}
	When the sLLF algorithm is used on instance $\hI$, for any EV $i \in \hS_2$ and $j \in \hF$, their laxity satisfy
	\begin{equation} 
		\ell_i(t) > \ell_j(t),\ t \in [ \max(a_i, a_j) ,\ t_-]
	\label{eq:z_order1} 
	\end{equation}

	\begin{proof}[\textbf{Proof}]

	\noindent
	By the construction of $\hS_2$, relation \eqref{eq:z_order1} holds at time $t = t_-$. By \cref{lemma:flip_conditions}, a necessary condition for the inequality in \eqref{eq:z_order1} to flip at some time $t+1 \leq t_-$ is for the second case of \eqref{eq:flip_cases} to hold for EV $i$. However, this condition cannot hold for any EV in $\hF$ or $\hS_1$. For EVs in $\hF$, condition $e_j(t+1) = 0$ in the second case of \eqref{eq:flip_cases} cannot happen because negative laxity at some time implies the energy demand will not be fulfilled. For EVs in $\hS_1$, \eqref{eq:z_order1} holds only after $e_j(t+1) = 0$ when they have energy demand fulfilled at time $t+1$. Consequently, condition \eqref{eq:z_order1} holds for all $t \in [ \max(a_i, a_j) , t_-]$.

	\end{proof}
\end{lemma}

\noindent
Here, the sLLF algorithm prioritizes EVs with smaller laxity so the presence of EVs with strictly greater laxity will not impact the charging of the EVs with smaller laxity. Defining $\tilde\hV=\hF \cup \hS_1$ as the set that does not contain the EVs in $\hS_2$, with the instances of this set denote as ${\tilde \hI}=\{a_i, d_i, e_{i}, \bar{r}_{i}; P(t)\}_{i\in\tilde\hV, t\in\hT}$, then the following can be obtained as a consequence of \cref{lemma:laxre}: 

\begin{corollary}
	\label{corollary:equ}
	Regardless of the actual instance being $\hI$ or $\tilde \hI$,  the EVs in $\tilde \hF$ are charged in exactly the same way under the sLLF algorithm by time $t_-$.
\end{corollary}

\noindent
The above condition for the sLLF algorithm being infeasible on some online feasible instances also holds for $\epsilon$ resource augmentation (both power and power+rate augmentation).

Now, consider comparing the sLLF algorithm with $\epsilon$ resource augmentation (either power or power+rate augmentation) and an offline algorithm. Let $\hI$ be an EV charging instance that are offline feasible and the sets $\hF$, $\hS_1$, and $\hS_2$ are defined under the sLLF algorithm. Since the EVs in $\hS_1$ are fully charged by time $t_-$ under both the sLLF algorithm with resource augmentation and the offline algorithm, we can have

\begin{equation}
	\label{eq:fes1}
	\Psi^\epsilon_{[0:t_-]} (\hS_1; \hI) = \Psi^*_{[0:t_-]} (\hS_1; \hI)
\end{equation}

\noindent
To be feasible, it is necessary for an algorithm to maintain $\ell_i(t) \geq 0,\ \forall t\in\hT$. Thus, for an EV $i \in \hF$, the offline algorithm must maintain $\ell_i(t_-) \geq 0$. Given that laxity is monotonically decreasing in the remaining energy demand, $e_i(t)$, the total energy fulfilled during the time interval $[0,t_-]$ under the offline algorithm must be strictly greater than that with the sLLF algorithm:

\begin{align*}
	& & & \Psi^\epsilon_{[0:t_-]} (\{i\}; \hI) & & < & & \Psi^*_{[0:t_-]} (\{i\}; \hI), & & i \in \hF \\
	\Rightarrow & & & \Psi^\epsilon_{[0:t_-]} (\hF; \hI) & & < & & \Psi^*_{[0:t_-]} (\hF; \hI) & & \eqnumber \label{eq:fes2} \\
	\intertext{for $\tilde \hV = \hV \backslash S_2$, together with \eqref{eq:fes1}, we have}
	& & & \Psi^\epsilon_{[0:t_-]} (\tilde\hV; \hI) & & < & & \Psi^*_{[0:t_-]} (\tilde\hV; \hI) \eqnumber \label{eq:fevt}
\end{align*}

\noindent
Additionally \cref{corollary:equ} implies

\begin{align*}
	& & & Psi^\epsilon_{[0:t_-]} ({i}; \hI) & & = & & \Psi^\epsilon_{[0:t_-]} ({i}; \tilde \hI), & & i \in \tilde \hV \\
	\Rightarrow & & & \Psi^\epsilon_{[0:t_-]} (\tilde \hV; \hI) & & = & & \Psi^\epsilon_{[0:t_-]} (\tilde \hV; \tilde \hI) \eqnumber
	\label{eq:fee}
\end{align*}

\noindent
Furthermore, because the charging instance $\hI$ is offline feasible then its sub-instance $\tilde\hI$ is also offline feasible. Similar to \cref{eq:fes1,eq:fes2,eq:fevt}, we can show that
 
\begin{align}
	& & & \Psi^\epsilon_{[0:t_-]} (\hS_1; \tilde\hI) & & = & & \Psi^*_{[0:t_-]} (\hS_1; \tilde \hI) \label{eq:fes1t} \\
	& & & \Psi^\epsilon_{[0:t_-]} (\hF; \tilde\hI) & & < & & \Psi^*_{[0:t_-]} (\hF; \tilde\hI) \label{eq:feft} \\
	& & & \Psi^\epsilon_{[0:t_-]} (\tilde\hV; \tilde \hI) & & < & & \Psi^*_{[0:t_-]} (\tilde\hV; \tilde\hI) \label{eq:fevtt}
\end{align}

\subsection{Proof of \cref{thm:power_aug_LLF}}
\label{apdxS:proof of th 1}

Consider the use of the sLLF algorithm on an offline feasible instance ${\hI}=\{a_i, d_i, e_{i}, \bar{r}_{i}; P(t)\}_{i\in\hV, t\in\hT}$ under $\epsilon$-power augmented resources. 

\begin{align}
	\intertext{Let}
	n = (1+\epsilon) \frac{ P_{\min} }{ P_{\max} }
	\label{eq:n_def}
\end{align}

\noindent
For $m \leq n$, we define the earliest time to charge at a power greater than $mP_{\max}$ for the rest of the time until $t_-$ as 

\begin{equation}
	\label{eq:t_i_define}
	\hspace{-3mm}t_{m} = \min \left\{ t \in \hT : \sum_{j \in V_t } \min ( \br_j , e_j(\tau)  )  \geq m P_{\max} , \tau \in [t , t_-] \right\}
\end{equation}

\noindent
Let $T_{m}= [t_{m-1} , t_{m} )$ and $\hat{T}_{m} =  [t_{m} , t_-] $ and denote their lengths by $|T_{m}|$ and $|\hat T_{m}|$. Also,

\begin{lemma}
	\label{lemma:TandD}
	For any integer $i \leq n-1$, the following two relations hold: 
	
	\begin{align}
		\label{eq:delta_5}
		\Psi^*_{[0:t_i]} (A_{T_i}; \tilde\hI) -  \Psi^{\epsilon}_{[0:t_i]} (A_{T_i}; \tilde\hI) & & > & & & P_{\max}  |\hat T_{i+1}| \\
		\label{eq:TandD}
		|T_{i}| & & > & & &|\hat T_{i+1}|
	\end{align}

	\begin{proof}[\textbf{Proof}]
	
	\noindent
	From \cref{eq:t_i_define} we can have

	\begin{align*}
		\sum_{ j \in V_{ ( t_{i - 1} ) - 1 } } \min ( \br_j , e_j( t_{i - 1} - 1 )  ) & & < & & (i-1)P_{\max}
	\end{align*}
	
	\noindent
	This implies that the EVs that have arrived before $t_{i-1}$ are charged at a total power of at most $(i-1)P_{\max}$ at $t_{i-1}$ and after. On the other hand, from \cref{eq:t_i_define}, the total power supply is at least $i P_{\max}$ during the interval $T_{i+1} = [t_{i} , t_{i+1}]$. Therefore, the total charging power to the EVs that arrive after $t_{i-1}$ is at least $P_{\max}$ during $T_{i+1}$. Since offline algorithm can only use a power of at most $P_{\max}$, for the EVs that arrive after $t_{i-1}$ we obtain

	\begin{equation}
		\label{eq:delta_1}
		\begin{aligned}
			& & & \Psi^*_{[0;t_{i+1}]} (A_{\hat T_{i-1}}; \tilde\hI) - \Psi^{\epsilon}_{[0;t_{i+1}]} (A_{\hat T_{i-1}}; \tilde\hI)\\
			< & & & \Psi^*_{[0;t_{i}]} (A_{\hat T_{i-1}}; \tilde\hI) - \Psi^{\epsilon}_{[0;t_{i}]} (A_{\hat T_{i-1}}; \tilde\hI)
		\end{aligned}
	\end{equation}
	
	\noindent
	The same argument can be applied to the interval $\hat T_{i+1} = [t_{i+1} , t_-]$. From \cref{eq:t_i_define}, the total charging power is at least $(i+1)  P_{\max}$ during $\hat T_{i+1}$. Therefore, during $\hat T_{i+1} $, the total charging power to the EVs that arrive after $t_{i-1}$ is at least $2 P_{\max}$. Since offline algorithm can only use a power of at most $P_{\max}$, the total energy supply to EVs in $\hat T_{i-1}$ under the augmented resources is greater than that without augmented resources, \ie:
 
	\begin{equation}
		\label{eq:delta_2}
		\begin{aligned}
			0 & & < & & &\Psi^*_{[0;t_{-}]} (A_{\hat T_{i-1}}; \tilde\hI) - \Psi^{\epsilon}_{[0;t_{-}]} (A_{\hat T_{i-1}}; \tilde\hI)\\
			& & < & & & \Psi^*_{[0;t_{i+1}]} (A_{\hat T_{i-1}}; \tilde\hI) - \Psi^{\epsilon}_{[0;t_{i+1}]} (A_{\hat T_{i-1}}; \tilde\hI) - P_{\max} |\hat T_{i+1}|
		\end{aligned}
	\end{equation}
	
	\noindent
	Combining \cref{eq:delta_1,eq:delta_2}, we have

	\begin{align*}
		& & 0 & & < & & & \Psi^*_{[0;t_{i+1}]} (A_{\hat T_{i-1}}; \tilde\hI) - \Psi^{\epsilon}_{[0;t_{i+1}]} (A_{\hat T_{i-1}}; \tilde\hI) \\
		& & & & & & & - P_{\max} |\hat T_{i+1}| \\
		\Leftrightarrow & & P_{\max} |\hat T_{i+1}| & & < & & & \Psi^*_{[0;t_{i+1}]} (A_{\hat T_{i-1}}; \tilde\hI) - \Psi^{\epsilon}_{[0;t_{i+1}]} (A_{\hat T_{i-1}}; \tilde\hI) \\
		& & & & < & & & \Psi^*_{[0;t_{i}]} (A_{\hat T_{i-1}}; \tilde\hI) - \Psi^{\epsilon}_{[0;t_{i}]} (A_{\hat T_{i-1}}; \tilde\hI) \\
		\intertext{since the set $A_{T_{i}}$ is identical to the subset of $A_{\hat T_{i-1}}$ that contains only the EVs that have arrived by $t_i$,} \\
		& & P_{\max} |\hat T_{i+1}| & & < & & & \Psi^*_{[0;t_{i}]} (A_{\hat T_{i-1}}; \tilde\hI) - \Psi^{\epsilon}_{[0;t_{i}]} (A_{\hat T_{i-1}}; \tilde\hI) \\
		& & & & = & & & \Psi^*_{[0;t_{i}]} (A_{\hat T_{i-1}}; \tilde\hI) - \Psi^{\epsilon}_{[0;t_{i}]} (A_{\hat T_{i-1}}; \tilde\hI) \\
		\intertext{therefore \cref{eq:delta_5}.}
	\end{align*}
	
	\noindent
	Finally, as all EVs in $A_{T_i}$ arrives after $t_{i-1}$, during $T_i$ the offline algorithm can charge a total energy of at most 
	$|T_i| P_{\max}$. Thus,
	\begin{align*}
		& & P_{\max} |T_i| & & \geq & & & \Psi^*_{[0:t_i]} (A_{T_i}; \tilde\hI) -  \Psi^{\epsilon}_{[0:t_i]} (A_{T_i}; \tilde\hI) \\
		\intertext{where from \cref{eq:delta_5}}
		& & P_{\max} |T_i| & & \geq & & & \Psi^*_{[0:t_i]} (A_{T_i}; \tilde\hI) -  \Psi^{\epsilon}_{[0:t_i]} (A_{T_i}; \tilde\hI) \\
		& & & & > & & & P_{\max} |\hat T_{i+1}| \\
		\Rightarrow & & P_{\max} |T_i| & & > & & & P_{\max} |\hat T_{i+1}| \\
		\Leftrightarrow & & |T_i| & & > & & & |\hat T_{i+1}| \\
		\intertext{therefore \cref{eq:TandD}.}
	\end{align*}
 
	\end{proof}

\end{lemma}

\begin{proof}[\textbf{Proof (\cref{thm:power_aug_LLF})}]

Suppose that there exists an offline feasible instance ${\hI}=\{a_i, d_i, e_{i},$ $ \bar{r}_{i}; P(t)\}_{i\in\hV, t\in\hT}$ such that the sLLF algorithm is not feasible with $\epsilon$-power augmented resources. Then, from the infeasibility condition of the sLLF algorithm defined previously, there exists another offline feasible instance ${\tilde \hI}=\{a_i, d_i, e_{i}, \bar{r}_{i}; P(t)\}_{i\in\tilde \hV, t\in\hT}$ such that \cref{eq:fevtt}. When $m=1$, we obtain $\sum_{ j \in V_{t_1 - 1} } \min ( \br_j , e_j(t_1 - 1)  )  < P_{\max}$. Let $\hS = \{ i \in A_{T_1} : e_i( t_1 ) > 0 \} \subset A_{T_1}$ denotes the set of EVs that arrive during $T_1$ and have not yet been fully charged by $t_1$. Because the number of EVs is upper bounded by $\frac{P_{\max}}{\br_{\min}}$ (from \cref{eq:t_i_define}) and the EVs in $A_{T_1} \backslash S$ are all fully charged, then
 
\begin{align*}
	P_{\max} |\hat T_2 | & & \leq & & & \Psi^*_{[0:t_-]} (A_{T_1}; \tilde\hI) -  \Psi^{\epsilon}_{[0:t_-]} (A_{T_1}; \tilde\hI)\\
	& & = & & & \Psi^*_{[0:t_-]} (\hS; \tilde\hI) -  \Psi^{\epsilon}_{[0:t_-]} (\hS; \tilde\hI)\\
	& & \leq & & & X \frac{P_{\max}}{\br_{\min}}
\end{align*}

\noindent
This leads to 

\begin{equation}
	\label{eq:D_bound1}
	|\hat T_2|   < \frac{ X }{ \br_{\min}  }  .
\end{equation}

\noindent
At time $t < t_{m - 1}$, we have

\begin{align*}
	\sum_{ j \in V_{ t_{m - 1} - 1 } } \min ( \br_j , e_j( t_{m - 1} - 1)  )  & & < & & (m- 1)P_{\max}
\end{align*}

\noindent
which implies that there are at most $(m- 1)  \frac{P_{\max}}{\br_{\min}}$ EVs with unfulfilled energy demand by time $t_f$. Meanwhile, at time $t \geq t_m$, we have

\begin{align*}
	\sum_{ j \in V_{t_{m} } } \min ( \br_j , e_j(t_{m})  ) & & \geq & & m P_{\max}
\end{align*}

\noindent
which implies that there are at least $m \frac{P_{\max}}{\br_{\max}}$ EVs with unfulfilled energy demand during $T_{m-1}$. Therefore, the number of EVs that arrive during $[t_{m-1} , t_{m}]$ is greater than the following:

\begin{equation}
	\label{eq:D_bound1_2}
	\frac{ m P_{\max} }{ \br_{\max} } - \frac{ (m- 1)  P_{\max} }{ \br_{\min} } \geq \frac{P_{\max} }{ \br_{\min} }.
\end{equation} 

\noindent
Since the inter-arrival periods of EVs are at least $N$, the length of $\hat T_{m-1}$ satisfies

\begin{equation}
	\label{eq:Dbound_2}
	|\hat T_{m-1}| \geq  \frac{P_{\max}N }{ \br_{\min} }  .
\end{equation}

\noindent
Consider the following recursion: 

\begin{align*}
|\hat T_{2}| & & = & & &  |\hat T_{3}|+  |T_{3}| \\
 & & \geq & & & |\hat T_{3}|+  |\hat T_{4}| & & \geq & & 2  |\hat T_{4}|+  |\hat T_{5}| \\
& & \geq & & & 3  |\hat T_{5}|+ 2 |\hat T_{6}| & & \geq & & 5  |\hat T_{6}|+ 3 |\hat T_{7}| \\
& & \geq & & & \cdots  & & \geq & & f_{k-2} |\hat T_{m-1}| + f_{k-3} |\hat T_{m}| 
\end{align*}

\noindent
where $f_k$ is the Fibonacci sequence defined by $ f_1 = 1$, $f_2 = 1$ and $f_{k} = f_{k-1} + f_{k-2}$ for $k\geq 3$. Thus,
 
\begin{equation*}
	 |\hat T_{2}|   > f_{m-2}  |\hat T_{m-1}|.
\end{equation*}

\noindent
Combining \cref{eq:D_bound1,eq:D_bound1_2,eq:Dbound_2} gives
 
\begin{equation}
	\label{eq:XN_ineq}
	\frac{ X }{ \br_{\min} } \ > \ |T_{2} | \ > \ f_{m-2}  |\hat T_{m-1}| \ > \ f_{m-2}  \frac{P_{\max} N}{ \br_{\min}  } 
\end{equation}

\noindent
From $m \leq n$ for $n$ defined in \cref{eq:n_def}, we obtain

\begin{align*}
	 \left \lfloor (1 + \epsilon) \frac{ P_{\min} }{ P_{\max} } \right\rfloor - 2  & & = & & & m-2  \\
	& & = & & & \log_\varphi \left(  \sqrt{ 5 }f_{n-2} + \frac{1}{2} \right) \\
	& & < & & & \log_\varphi \left(   \frac{ \sqrt{ 5 } X } { N P_{\max}  } + \frac{1}{2} \right)
\end{align*}

\noindent
which gives \cref{thm:power_aug_LLF}

\end{proof}

\subsection{Proof of \cref{thm:achievable_rate}}
\label{apdxS:proof of th 2}

\begin{proof}[\textbf{Proof (\cref{thm:achievable_rate})}]

Suppose that there exists an instance ${\hI}=\{a_i, d_i, e_{i}, \bar{r}_{i}; P(t)\}_{i\in\hV, t\in\hT}$ such that the sLLF algorithm is not feasible with $\epsilon$-power+rate augmentation. We then have \cref{eq:fevtt}, for another instance ${\tilde\hI}=\{a_i, d_i, e_{i}, \bar{r}_{i}; P(t)\}_{i\in\tilde \hV, t\in\hT}$. Let $\hS(\tilde \hV)$ be the set of EVs in the instance $\tilde \hI$ that receive strictly less energy under online algorithm than under offline algorithm by some time $t$ at which $\Psi^\epsilon_{[0:t]} ( \tilde \hV; \tilde\hI ) <   \Psi^*_{[0:t]} (\tilde  \hV; \tilde \hI )$:

\begin{align*}
	\hS(\tilde \hV) = \big\{ i  \in \tilde \hV: \exists t\in\hT & & \text{s.t.} & & & \Psi^\epsilon_{[0:t]} ( \{i\};  \tilde \hI ) <  \Psi^*_{[0:t]} ( \{i\};  \tilde \hI ) \\
	& & \& & & & \Psi^\epsilon_{[0:t]} ( \tilde \hV;   \tilde \hI) <  \Psi^*_{[0:t]} ( \tilde \hV;  \tilde \hI ) \big\}
\end{align*} 

\noindent
In view of \cref{eq:fevtt}, $\hS(\tilde \hV) \neq \emptyset$. Consider EV $j$ that arrives earliest among those in $\hS(\tilde \hV)$. There exists a time $t \in [a_j, d_j]$ such that

\begin{align*}
	& & & \Psi^\epsilon_{[0:t]}(\{j\}; \tilde\hI) & & < & & \Psi^*_{[0:t]}(\{j\}; \tilde\hI), & j \in \tilde\hV \\
	\Rightarrow & & & \Psi^\epsilon_{[0:t]}(\tilde\hV; \tilde\hI) & & < & & \Psi^*_{[0:t]}(\tilde\hV; \tilde\hI)
\end{align*}

\noindent
Notice that $\Psi^\epsilon_{[0:a_j - 1]}(\tilde \hV; \tilde \hI ) < \Psi^*_{[0:a_j - 1]}(\tilde \hV; \tilde \hI)$ can only happen when there is another EV in $\hS(\tilde\hV)$ that arrives before EV $j$, which, however, contradicts the definitions of $\hS(\tilde \hV)$ and $j$. Thus, 

\begin{align*}
	\Psi^\epsilon_{[0:a_j - 1]}(\tilde \hV; \tilde\hI) & & \geq & & & \Psi^*_{[0:a_j - 1]}(\tilde \hV; \tilde\hI) \\
	\intertext{which implies} \\
	\Psi_{[a_j:t]}^\epsilon (\tilde \hV; \tilde\hI) & & < & & & \Psi^*_{[a_j:t]}(\tilde \hV; \tilde\hI) \eqnumber \label{eq:S_tot} 
\end{align*}

\noindent
Now, let us take a look at the energy demand fulfilled during the interval $[a_j, t]$ under the sLLF algorithm with $\epsilon$-power+rate augmentation. Define the over-loaded times as

\begin{equation*}
	T_o = \left\{ t \in [a_j, t] : \sum_{i \in \tilde \hV} r_i(t) = (1+\epsilon) P(t) \right\}
\end{equation*}

\noindent
and under-loaded times as

\begin{equation*}
	T_u = \left\{ t \in [a_j, t]: \sum_{i \in \tilde \hV} r_i(t) < (1+\epsilon) P(t) \right\} 
\end{equation*}

\noindent
we have $|T_o|+ |T_u| = t +1 - a_j$. The total energy demand fulfilled during the over-loaded period is lower bounded by $|T_o| (1+\epsilon) \min_{\tau \in [a_j , d_j]}P(\tau)$, while during the under-loaded period is at least $|T_u|  (1+\epsilon) \br_j$. Hence, the total and individual energy demands fulfilled during $[a_j, t]$ are lower bounded by
 
\begin{align}
	\label{eq:energy_rel_1}
	(1+\epsilon) \left( |T_u|   \br_j + |T_o|  \min_{\tau \in [a_j , d_j]}P(\tau) \right) & & \leq & & \Psi^\epsilon_{[a_j:t]} (\tilde \hV; \tilde \hI) \\
	\label{eq:energy_rel_4}
	(1+\epsilon) |T_u|  \br_j  & & \leq & & \Psi^\epsilon_{[a_j:t]} (\{j\}; \tilde \hI) 
\end{align}

\noindent
Next, let us take a look at the energy demand fulfilled during the interval $[a_j, t+1]$ by offline algorithm without resource augmentation. The total energy fulfilled is upper bounded by

\begin{equation}
	\Psi^*_{[a_j:t]} (\tilde \hV; \tilde \hI) \leq (t+1-a_j) \max_{ \tau \in [a_j , d_j] } P(\tau) \eqnumber \label{eq:energy_rel_2}
\end{equation}

\noindent
while the energy fulfilled to individual EV $j$ is upper bounded by

\begin{equation}
	\label{eq:energy_rel_3} 
	\Psi^*_{a_j:t} (j)  \leq (t+1-a_j) \br_j 
\end{equation}

\noindent
From \cref{eq:S_tot,eq:energy_rel_4,eq:energy_rel_3}, we have

\begin{align*}
	& & & (1+\epsilon) |T_u|  \br_j  & & \leq & & \Psi^\epsilon_{[a_j:t]} (\{j\}; \tilde \hI) \\
	& & & & & < & & \Psi^*_{a_j:t} (j) \\
	& & & & & \leq & & (t+1-a_j) \br_j \\
	\Rightarrow & & & |T_u| (1+\epsilon) & & < & & (t-a_j+1) \eqnumber \label{eq:rate_eq1} 
\end{align*} 

\noindent
while from \cref{eq:S_tot,eq:energy_rel_1,eq:energy_rel_2}, we have

\begin{align*}  
	& & & (1+\epsilon) \left( |T_u|   \br_j + |T_o|  \min_{\tau \in [a_j , d_j]}P(\tau) \right) \\
	& \leq & & \Psi^\epsilon_{[a_j:t]} (\tilde \hV; \tilde \hI) \\
	& < & & \Psi^*_{a_j:t} (j) \\
	& \leq & & (t+1-a_j) \max_{ \tau \in [a_j , d_j] } P(\tau)
\end{align*}

\noindent
thus,

\begin{equation}
	\label{eq:rate_eq2}
	(1+\epsilon) \left( |T_u |\br_j + |T_o| \min_{ \tau \in [a_j , d_j] } P(\tau) \right) < (t+1-a_j) \max_{ \tau \in [a_j , d_j] } P(\tau)
\end{equation}

\noindent
Combining both \cref{eq:rate_eq1,eq:rate_eq2},

\begin{align*}
	& & & (|T_u | + |T_o|) (1+\epsilon)  \min_{ \tau \in [a_j , d_j] } P(\tau) \\
	& < & & (t - a_j+1) (  \max_{ \tau \in [a_j , d_j] } P(\tau)  +  \min_{ \tau \in [a_j , d_j] } P(\tau)  -  \br_j) 
\end{align*}

\noindent
while $|T_o| + |T_u| = t +1 - a_j$. Thus,
 
\begin{equation*}
	\epsilon < \max_{ \tau_1 , \tau_2 \in [ a_j, d_j ]} \frac{ P(\tau_1) }{ P (\tau_2) }  - \min_{i \in \hV}  \max_{\tau \in [a_i, d_i]} \frac{  \br_i }{ P (\tau) }
\end{equation*}

\end{proof}

%
%
%


\ifCLASSOPTIONcaptionsoff
  \newpage
\fi



%


\bibliographystyle{IEEEtran}
\bibliography{sigproc}

\begin{thebibliography}{10}
\providecommand{\url}[1]{#1}
\csname url@samestyle\endcsname
\providecommand{\newblock}{\relax}
\providecommand{\bibinfo}[2]{#2}
\providecommand{\BIBentrySTDinterwordspacing}{\spaceskip=0pt\relax}
\providecommand{\BIBentryALTinterwordstretchfactor}{4}
\providecommand{\BIBentryALTinterwordspacing}{\spaceskip=\fontdimen2\font plus
\BIBentryALTinterwordstretchfactor\fontdimen3\font minus
  \fontdimen4\font\relax}
\providecommand{\BIBforeignlanguage}[2]{{%
\expandafter\ifx\csname l@#1\endcsname\relax
\typeout{** WARNING: IEEEtran.bst: No hyphenation pattern has been}%
\typeout{** loaded for the language `#1'. Using the pattern for}%
\typeout{** the default language instead.}%
\else
\language=\csname l@#1\endcsname
\fi
#2}}
\providecommand{\BIBdecl}{\relax}
\BIBdecl

\bibitem{EVoutlook}
I.~E. Agency, ``{Global EV outlook 2020},''
  \url{https://www.iea.org/reports/global-ev-outlook-2020}, 2020.

\bibitem{sundstrom2010planning}
O.~Sundstr{\"o}m and C.~Binding, ``Planning electric-drive vehicle charging
  under constrained grid conditions,'' in \emph{Power System Technology
  (POWERCON), 2010 International Conference on}.\hskip 1em plus 0.5em minus
  0.4em\relax IEEE, 2010, pp. 1--6.

\bibitem{6670091}
A.~Subramanian, M.~J. Garcia, D.~S. Callaway, K.~Poolla, and P.~Varaiya,
  ``Real-time scheduling of distributed resources,'' \emph{IEEE Transactions on
  Smart Grid}, vol.~4, no.~4, pp. 2122--2130, 2013.

\bibitem{richardson2012local}
P.~Richardson, D.~Flynn, and A.~Keane, ``Local versus centralized charging
  strategies for electric vehicles in low voltage distribution systems,''
  \emph{IEEE Transactions on Smart Grid}, vol.~3, no.~2, pp. 1020--1028, 2012.

\bibitem{gan2013optimal}
L.~Gan, U.~Topcu, and S.~H. Low, ``Optimal decentralized protocol for electric
  vehicle charging,'' \emph{IEEE Transactions on Power Systems}, vol.~28,
  no.~2, pp. 940--951, 2013.

\bibitem{ma2010decentralized}
Z.~Ma, D.~Callaway, and I.~Hiskens, ``Decentralized charging control for large
  populations of plug-in electric vehicles,'' in \emph{49th IEEE conference on
  decision and control (CDC)}.\hskip 1em plus 0.5em minus 0.4em\relax IEEE,
  2010, pp. 206--212.

\bibitem{chen2014electric}
N.~Chen, C.~W. Tan, and T.~Q. Quek, ``Electric vehicle charging in smart grid:
  Optimality and valley-filling algorithms,'' \emph{IEEE Journal of Selected
  Topics in Signal Processing}, vol.~8, no.~6, pp. 1073--1083, 2014.

\bibitem{gan2013real}
L.~Gan, A.~Wierman, U.~Topcu, N.~Chen, and S.~H. Low, ``Real-time deferrable
  load control: handling the uncertainties of renewable generation,'' in
  \emph{Proceedings of the fourth international conference on Future energy
  systems}.\hskip 1em plus 0.5em minus 0.4em\relax ACM, 2013, pp. 113--124.

\bibitem{chen2012large}
S.~Chen, Y.~Ji, and L.~Tong, ``Large scale charging of electric vehicles,'' in
  \emph{2012 IEEE Power and Energy Society General Meeting}.\hskip 1em plus
  0.5em minus 0.4em\relax IEEE, 2012, pp. 1--9.

\bibitem{yu2016deadline}
Z.~Yu, Y.~Xu, and L.~Tong, ``Deadline scheduling as restless bandits,''
  \emph{arXiv preprint arXiv:1610.00399}, 2016.

\bibitem{Stein:2012:MOM:2343776.2343792}
\BIBentryALTinterwordspacing
S.~Stein, E.~Gerding, V.~Robu, and N.~R. Jennings, ``A model-based online
  mechanism with pre-commitment and its application to electric vehicle
  charging,'' in \emph{Proceedings of the 11th International Conference on
  Autonomous Agents and Multiagent Systems - Volume 2}, ser. AAMAS '12.\hskip
  1em plus 0.5em minus 0.4em\relax Richland, SC: International Foundation for
  Autonomous Agents and Multiagent Systems, 2012, pp. 669--676. [Online].
  Available: \url{http://dl.acm.org/citation.cfm?id=2343776.2343792}
\BIBentrySTDinterwordspacing

\bibitem{sha2004real}
L.~Sha, T.~Abdelzaher, K.-E. {\AA}rz{\'e}n, A.~Cervin, T.~Baker, A.~Burns,
  G.~Buttazzo, M.~Caccamo, J.~Lehoczky, and A.~K. Mok, ``Real time scheduling
  theory: A historical perspective,'' \emph{Real-time systems}, vol.~28, no.
  2-3, pp. 101--155, 2004.

\bibitem{carvalho2015congestion}
R.~Carvalho, L.~Buzna, R.~Gibbens, and F.~Kelly, ``Congestion control in
  charging of electric vehicles,'' \emph{arXiv preprint arXiv:1501.06957},
  2015.

\bibitem{7809844}
F.~Kong, Q.~Xiang, L.~Kong, and X.~Liu, ``On-line event-driven scheduling for
  electric vehicle charging via park-and-charge,'' in \emph{2016 IEEE Real-Time
  Systems Symposium (RTSS)2016 IEEE}.\hskip 1em plus 0.5em minus 0.4em\relax
  IEEE, 2016, pp. 69--78.

\bibitem{Guo}
L.~{Guo}, K.~F. {Erliksson}, and S.~H. {Low}, ``Optimal online adaptive
  electric vehicle charging,'' in \emph{2017 IEEE Power Energy Society General
  Meeting}, 2017.

\bibitem{AlonsoOptimalAlgorithm2014}
\BIBentryALTinterwordspacing
M.~Alonso, H.~Amaris, J.~G. Germain, and J.~M. Galan, ``Optimal charging
  scheduling of electric vehicles in smart grids by heuristic algorithms,''
  \emph{Energies}, vol.~7, no.~4, pp. 2449--2475, 2014. [Online]. Available:
  \url{https://www.mdpi.com/1996-1073/7/4/2449}
\BIBentrySTDinterwordspacing

\bibitem{MaoIntelligentModes2019}
\BIBentryALTinterwordspacing
T.~Mao, X.~Zhang, and B.~Zhou, ``Intelligent energy management algorithms for
  ev-charging scheduling with consideration of multiple ev charging modes,''
  \emph{Energies}, vol.~12, no.~2, 2019. [Online]. Available:
  \url{https://www.mdpi.com/1996-1073/12/2/265}
\BIBentrySTDinterwordspacing

\bibitem{ElmehdiGeneticFleet2019}
\BIBentryALTinterwordspacing
M.~Elmehdi and M.~Abdelilah, ``Genetic algorithm for optimal charge scheduling
  of electric vehicle fleet,'' in \emph{Proceedings of the 2nd International
  Conference on Networking, Information Systems \& Security}, ser.
  NISS19.\hskip 1em plus 0.5em minus 0.4em\relax New York, NY, USA: Association
  for Computing Machinery, 2019. [Online]. Available:
  \url{https://doi.org/10.1145/3320326.3320329}
\BIBentrySTDinterwordspacing

\bibitem{stankovic2012deadline}
J.~A. Stankovic, M.~Spuri, K.~Ramamritham, and G.~C. Buttazzo, \emph{Deadline
  scheduling for real-time systems: EDF and related algorithms}.\hskip 1em plus
  0.5em minus 0.4em\relax Springer Science \& Business Media, 2012, vol. 460.

\bibitem{whittle1988restless}
P.~Whittle, ``Restless bandits: Activity allocation in a changing world,''
  \emph{Journal of applied probability}, vol.~25, no.~A, pp. 287--298, 1988.

\bibitem{BrennaElectricEstimation2020}
\BIBentryALTinterwordspacing
M.~Brenna, F.~Foiadelli, C.~Leone, and M.~Longo, ``Electric vehicles charging
  technology review and optimal size estimation,'' \emph{Journal of Electrical
  Engineering {\&} Technology}, vol.~15, no.~6, pp. 2539--2552, Nov 2020.
  [Online]. Available: \url{https://doi.org/10.1007/s42835-020-00547-x}
\BIBentrySTDinterwordspacing

\bibitem{AmjadATechniques2018}
\BIBentryALTinterwordspacing
M.~Amjad, A.~Ahmad, M.~H. Rehmani, and T.~Umer, ``A review of evs charging:
  From the perspective of energy optimization, optimization approaches, and
  charging techniques,'' \emph{Transportation Research Part D: Transport and
  Environment}, vol.~62, pp. 386 -- 417, 2018. [Online]. Available:
  \url{http://www.sciencedirect.com/science/article/pii/S1361920917306120}
\BIBentrySTDinterwordspacing

\bibitem{kalyanasundaram2000speed}
B.~Kalyanasundaram and K.~Pruhs, ``Speed is as powerful as clairvoyance,''
  \emph{Journal of the ACM (JACM)}, vol.~47, no.~4, pp. 617--643, 2000.

\bibitem{Phillips2002}
\BIBentryALTinterwordspacing
Phillips, Stein, Torng, and Wein, ``Optimal time-critical scheduling via
  resource augmentation,'' \emph{Algorithmica}, vol.~32, no.~2, pp. 163--200,
  2002. [Online]. Available: \url{http://dx.doi.org/10.1007/s00453-001-0068-9}
\BIBentrySTDinterwordspacing

\bibitem{im2014competitive}
S.~Im, J.~Kulkarni, and K.~Munagala, ``Competitive algorithms from competitive
  equilibria: non-clairvoyant scheduling under polyhedral constraints,'' in
  \emph{Proceedings of the 46th Annual ACM Symposium on Theory of
  Computing}.\hskip 1em plus 0.5em minus 0.4em\relax ACM, 2014, pp. 313--322.

\bibitem{im2014selfishmigrate}
S.~Im, J.~Kulkarni, K.~Munagala, and K.~Pruhs, ``Selfishmigrate: A scalable
  algorithm for non-clairvoyantly scheduling heterogeneous processors,'' in
  \emph{Foundations of Computer Science (FOCS), 2014 IEEE 55th Annual Symposium
  on}.\hskip 1em plus 0.5em minus 0.4em\relax IEEE, 2014, pp. 531--540.

\bibitem{liu1973scheduling}
C.~L. Liu and J.~W. Layland, ``Scheduling algorithms for multiprogramming in a
  hard-real-time environment,'' \emph{Journal of the ACM (JACM)}, vol.~20,
  no.~1, pp. 46--61, 1973.

\bibitem{dertouzos1989multiprocessor}
M.~L. Dertouzos and A.~K. Mok, ``Multiprocessor online scheduling of
  hard-real-time tasks,'' \emph{IEEE Transactions on software engineering},
  vol.~15, no.~12, pp. 1497--1506, 1989.

\bibitem{davis2011survey}
R.~I. Davis and A.~Burns, ``A survey of hard real-time scheduling for
  multiprocessor systems,'' \emph{ACM Computing Surveys (CSUR)}, vol.~43,
  no.~4, p.~35, 2011.

\bibitem{Lee2016}
G.~Lee, T.~Lee, Z.~Low, S.~H. Low, and C.~Ortega, ``Adaptive charging network
  for electric vehicles,'' in \emph{Proceedings of the IEEE Global Conference
  on Signal and Information Processing (GlobalSIP)}, Washington, DC, December
  2016.

\bibitem{ZLee2018}
Z.~Lee, D.~Chang, C.~Jin, G.~S. Lee, R.~Lee, T.~Lee, and S.~H. Low,
  ``Large-scale adaptive electric vehicle charging,'' in \emph{Proceedings of
  the IEEE SmartGridComm Conference}, Aalborg, Denmark, October 2018.

\bibitem{LeeLiLow2019a}
Z.~Lee, T.~Li, and S.~H. Low, ``{ACN-Data}: Analysis and applications of an
  open {EV} charging dataset,'' in \emph{Proc. Tenth ACM International
  Conference on Future Energy Systems (e-Energy)}, Phoenix, AZ, June 2019.

\bibitem{Mok1983FundamentalEnvironment}
\BIBentryALTinterwordspacing
A.~K.-L. Mok, ``{Fundamental Design Problems of Distributed Systems for The
  Hard-Real-Time Environment},'' Ph.D. dissertation, Massachusetts Institute of
  Technology, 1983. [Online]. Available:
  \url{http://hdl.handle.net/1721.1/15670}
\BIBentrySTDinterwordspacing

\bibitem{Stankovic:1998:DSR:552538}
J.~A. Stankovic, K.~Ramamritham, and M.~Spuri, \emph{Deadline Scheduling for
  Real-Time Systems: EDF and Related Algorithms}.\hskip 1em plus 0.5em minus
  0.4em\relax Norwell, MA, USA: Kluwer Academic Publishers, 1998.

\bibitem{Liu1973SchedulingEnvironment}
\BIBentryALTinterwordspacing
C.~L. Liu and J.~W. Layland, ``Scheduling algorithms for multiprogramming in a
  hard-real-time environment,'' \emph{J. ACM}, vol.~20, no.~1, p. 46–61, Jan.
  1973. [Online]. Available: \url{https://doi.org/10.1145/321738.321743}
\BIBentrySTDinterwordspacing

\bibitem{Kay1988AScheduler}
\BIBentryALTinterwordspacing
J.~Kay and P.~Lauder, ``A fair share scheduler,'' \emph{Commun. ACM}, vol.~31,
  no.~1, p. 44–55, Jan. 1988. [Online]. Available:
  \url{https://doi.org/10.1145/35043.35047}
\BIBentrySTDinterwordspacing

\bibitem{Jiang2020FasterDM}
S.~Jiang, Z.~Song, O.~Weinstein, and H.~Zhang, ``Faster dynamic matrix inverse
  for faster lps,'' \emph{ArXiv}, vol. abs/2004.07470, 2020.

\end{thebibliography}

%








\end{document}